\documentclass[12pt,twoside]{article}
\usepackage{amsmath,amssymb,amsthm,psfrag}
\usepackage{epsfig,graphicx,verbatim}
\usepackage{color,psfrag}
\usepackage[bookmarks, bookmarksopen=true, bookmarksnumbered=true]{hyperref}
\numberwithin{equation}{section}
\newtheorem{thm}[equation]{Theorem}

\newtheorem{lem}[equation]{Lemma}
\newtheorem{defn}[equation]{Definition}
\newcounter{mycount}
\newenvironment{romlist}{\begin{list}{\rm(\roman{mycount})}%
   {\usecounter{mycount}\labelwidth=1cm\itemsep 0pt}}{\end{list}}

\newcommand{\tr}{\operatorname{tr}}

\newcommand{\uinvnorm}{|\kern-2pt|\kern-2pt|}

\def\1{\mbox{1\hskip-.25em l}}
\newcommand{\beq}{\begin{equation}}
\newcommand{\eeq}{\end{equation}}
\newcommand{\ZZ}{\mathbb{Z}}

\theoremstyle{plain}

\theoremstyle{definition}

\theoremstyle{remark}

\def\esssup{\mathrm{ess\,sup}\,}
\def\anb{a_{n,\b}}
\def\amb{a_{m,\b}}

\def\sF{\mathcal F}
\def\sG{\mathcal G}

\def\sC{\mathcal C}
\def\sH{\mathcal H}
\def\sT{\mathcal T}
\def\qq{\qquad}
\def\q{\quad}
\def\a{\alpha}
\def\b{\beta}
\def\de{\delta}
\def\lam{\lambda}
\def\th{\theta}
\def\s{\sigma}
\def\psidm{\psi_m}
\def\psidz{\psi_0}
\def\rc{random-cluster}

\def\eps{\epsilon}
\def\g{\gamma}
\def\Ga{\Gamma}
\def\Si{\Sigma}
\def\rc{random-cluster}
\def\ZZ{{\mathbb Z}}
\def\RR{{\mathbb R}}

\def\PP{{\mathbb P}}
\def\II{{\mathbb I}}
\def\CC{{\mathbb C}}
\def\PPP{\ol\PP} %^{\mathrm{P}}}
\def\ZR{\ZZ\times\RR}
\def\Plb{\PP_{\lam,\de}}
\def\PLlb{\PP_{\La,\lam,\de}}
\def\PLlbq{\PP_{\La,\lam,\de,q}}
\def\PLmblb{\PP_{\Lamb,\lam,\de}}
\def\Plbq{\PP_{\lam,\de,q}}
\def\Om{\Omega}
\def\OmL{\Omega_\La}
\def\Ommb{\Om_{m,\b}}
\def\Omnb{\Om_{n,\b}}
\def\om{\omega}
\def\De{\Delta}

\def\Smb{\Sigma_{m,\b}}
\def\th{\theta}
\def\eps{\epsilon}
\def\La{\Lambda}
\def\Lamb{\La_{m,\b}}
\def\Lanb{\La_{n,\b}}
\def\oo{\infty}
\def\lest{\le_{\mathrm {st}}}

\def\thetac{\theta_{\mathrm{c}}}
\def\kp{k^{\mathrm p}}
\def\kw{k^{\mathrm w}}
\def\kpp{k^{\mathrm {pp}}}
\def\pp{{\mathrm {pp}}}
\def\kpw{k^{\mathrm {pw}}}
\def\gest{\ge_{\mathrm {st}}}
\def\be{\begin{equation}}
\def\ee{\end{equation}}
\def\sm{\setminus}
\def\resp{respectively}
\def\pd{\partial}
\def\pdh{\pd^{\mathrm h}}
\def\lra{\leftrightarrow}
\def\nlra{\nleftrightarrow}
\def\es{\varnothing}
\def\phm{\phi_m}
\def\phL{\phi_\La}
\def\phmb{\phi_{m,\b}}

\def\pn{\phi_n}

\def\phnb{\phi_{n,\b}}
\def\ol#1{\overline{#1}}
\def\wh#1{\widehat{#1}}
\def\ophi{\ol\phi}
\def\O{{\mathrm O}}
\def\o{{\mathrm o}}
\def\Dt{\Delta t}

\begin{document}
\title{Entanglement in the quantum Ising model}
\author{Geoffrey R.\ Grimmett\footnote{Centre for Mathematical Sciences,
University of Cambridge, Wilberforce Road, Cambridge CB3 0WB, UK},
Tobias J.\ Osborne\footnote{Department of Mathematics, Royal
Holloway, University of London, Egham, Surrey TW20 0EX, UK},\\
 Petra F.\ Scudo\footnote{Scuola Internazionale Superiore di Studi Avanzati,
via Beirut 2--4, 34014 Trieste, Italy; INFN, Sezione di Trieste, 
Trieste, Italy}}
   \maketitle

\begin{abstract}
We study the asymptotic scaling of the entanglement of a block of
spins for the ground state of the one-dimensional quantum Ising
model with transverse field. When the field is
sufficiently strong, the entanglement grows at most logarithmically
in the number of spins. The proof utilises a transformation to a
model of classical probability called the continuum \rc\ model, and
is based on a property of the latter model termed ratio weak-mixing.
Our proof applies equally to a large class of disordered
interactions.
\end{abstract}

\section{The quantum Ising model}

The quantum Ising model in a transverse magnetic field is one of the
most famous examples of exactly solvable one-dimensional quantum
models. The solution was first given by Pfeuty in \cite{pfeuty},
based on earlier works by Lieb, Schultz, and Mattis \cite{lieb} and
by McCoy \cite{mccoy}. The diagonalisation of the Hamiltonian and
the determination of the energy eigenstates is based on methods
developed by Jordan and Wigner \cite{jw} in the theory of second
quantisation of fermion fields, and by Bogoliubov \cite{bog} in the
theory of superconductivity. This model exhibits a second-order
phase transition in the ground state when the temperature of the
system is zero. The existence of the phase transition and the
computation of the spin--spin correlation functions were studied in
\cite{pfeuty}; rigorous results for the correlation functions in the
presence of disorder are provided in \cite{AKN,CKP}.

Quantum systems, unlike classical systems, can support composite
pure states for which it is impossible to assign a definite state to
two or more subsystems. States with this property are known as
entangled states and have attracted a great deal of interest
recently due to their resource-like properties. The investigation of
the entanglement properties of strongly interacting quantum spin
systems, with a view toward quantum phase transitions, was initiated
by Osterloh et al.\ \cite{fazio} and by Osborne and Nielsen
\cite{on} (see, for example, \cite{amico} and the references therein for
further studies). It is now understood that the strength of quantum
entanglement is related to the number of parameters required to
describe a quantum state classically. Thus, for 1D systems, the
scaling of the \emph{geometric entropy} --- the degree of
entanglement of a distinguished subsystem with respect to the rest
---  has emerged as the crucial parameter which quantifies whether
the state is hard or easy to simulate \cite{verstraete}. It has been
conjectured that the entropy of entanglement obeys an area law,
scaling as the boundary area in the subcritical phase, with a
possible logarithmic correction for the critical phase. There is a
paucity of rigorous results concerning the scaling of the
entanglement of a block for the quantum Ising model; the above
results are typically obtained by numerical calculations, or
conformal field theory methods \cite{amico}. There are some
rigorous derivations of the scaling of the entropy function for
certain 1D spin models (specialised essentially to the $XY$ model),
see \cite{amico} for further references.

In this paper, we utilise a new method for studying the
entanglement properties of the quantum Ising model. This is
based on a representation formulated by
Aizenman, Klein, and Newman \cite{AKN} 
of the model in terms of a continuum random-cluster
model on a certain space--time graph. (See also
the earlier paper \cite{CKP}.) Using a technique termed
ratio weak-mixing, developed by Alexander \cite{Al1,Al2} for
random-cluster and Potts models on discrete lattices, we prove a bound on the
entanglement entropy in the subcritical regime, when the magnetic
field intensity is strong compared to the spin coupling.

The quantum Ising model is defined as follows.
 Let $L \ge 0$. For $m \ge 0$, let
$\Delta_m= \{-m,-m+1,\dots, m+L\}$ be a subset of the
one-dimensional lattice $\ZZ$, and attach to each vertex $x\in
\Delta_m$ a quantum spin-$\frac12$ with local Hilbert space
$\CC^2$. The Hilbert space $\mathcal{H}$ for the system is
$\mathcal{H} = \bigotimes_{x=-m}^{m+L} \CC^2$. A convenient basis
for each spin is provided by the two eigenstates
$|+\rangle=\left(\begin{matrix} 1\\0\end{matrix}\right)$,
$|-\rangle=\left(\begin{matrix}0\\1\end{matrix}\right)$, of the Pauli operator
$$
\sigma^{(3)}_x = \left(
\begin{array}{cc} 1 & 0
\\ 0 & -1\end{array} \right),
$$
at the site $x$, corresponding to the eigenvalues $\pm 1$.
The other two Pauli operators with respect
to this basis are represented by the matrices
\begin{equation}
\sigma^{(1)}_x= \left( \begin{array}{cc} 0 & 1 \\ 1 & 0\end{array}
\right), \qquad \sigma^{(2)}_x= \left ( \begin{array}{cc} 0& -i \\
i & 0\end{array}\right).
\end{equation}
A complete basis for $\mathcal{H}$ is given by the tensor products
(over $x$) of the eigenstates of $\sigma^{(3)}_x$. In
the following, $|\phi\rangle$ denotes a vector and $\langle \phi|$
its adjoint. As a notational convenience in this paper, we shall represent sub-intervals
of $\ZZ$ as real intervals, writing for example $\De_m=[-m, m+L]$.

The spins in $\Delta_m$ interact via the quantum Ising Hamiltonian
\begin{equation}\label{ham}
H_{m} = -\frac{1}{2} \sum_{\langle x,
y\rangle}\lambda_{x,y}\sigma^{(3)}_x
 \sigma^{(3)}_y -  \sum_{x} \delta_x\sigma^{(1)}_x,
\end{equation}
generating the operator $e^{-\b H_m}$ where $\b$ denotes inverse
temperature. Here,  $\lambda_{x,y}\geq 0$ and $\delta_x\geq 0$ are
the spin-coupling and external-field intensities, respectively, and
$\sum_{\langle x, y\rangle}$ denotes a sum over all (distinct)
unordered pairs of spins. We concentrate here on the
case of interactions between neighbouring spins: $\lambda_{x,y} = 0$
for $|x-y|\ge 2$. While we shall phrase our results for the
translation-invariant case $\lambda_{x,x+1} = \lambda$ and $\delta_x
= \delta$, our approach can be extended to random couplings
satisfying the condition
\begin{equation}
\mathbb{P}(\lambda_{x,y} < \lambda) = \mathbb{P}(\delta_x > \delta)
= 1,
\end{equation}
with $\theta\equiv \lambda/\delta$ a sufficiently small constant
(see Section \ref{sec:disorder}).
The ensuing Hamiltonian has a unique pure ground state $|\psidm
\rangle$ defined at $T=0$ ($\b\to\oo$) as the eigenvector
corresponding to the lowest eigenvalue of $H_m$. In the
translation-invariant case the ground state $|\psidm\rangle$ depends
only on the ratio $\theta$.

For definiteness, we shall work here with a free boundary condition
on $\De_m$, but we note that the same methods are valid with a
periodic (or wired) boundary condition, in which $\De_m$ is embedded
on a circle. One difference worthy of note is that the correlation
functions of the \emph{critical} model are expected to depend on the choice
of boundary conditions, see \cite{pfeuty}.

We write $\rho_m(\b)=e^{-\b H_m}/\tr(e^{-\b H_m})$, and
$$
\rho_m=\lim_{\b\to\oo}\rho_m(\b) =|\psidm \rangle\langle\psidm|
$$
for the density operator corresponding to the ground state of the
system. The existence of the limit follows by \rc\ methods, see
\cite{AKN}, and we return to this in Section \ref{sec:rcrep}. The
ground-state entanglement of $|\psidm\rangle$ is quantified by
partitioning the spin chain $\Delta_m$ into two disjoint sets $[0,
L]$ and $\Delta_m\setminus [0, L]$ and by considering the entropy of
the \emph{reduced density operator} \be\label{reddo} \rho_m^L =
\tr_{\Delta_m\setminus [0, L]}(|\psidm\rangle\langle \psidm|). \ee
One may similarly define, for finite $\b$, the reduced operator
$\rho_m^L(\b)$. In both cases, the trace is performed over the
Hilbert space $(\bigotimes_{x=-m}^{-1}\mathbb{C}^2)\otimes(
\bigotimes_{x=L+1}^{m+L} \CC^2)$ of the spins belonging to
$\De_m\sm[0,L]$. Note that $\rho_m^L$ is a positive semi-definite
operator on the Hilbert space $\sH_L$ of dimension $d= 2^{L+1}$ of
spins indexed by the interval $[0, L]$. By the spectral theorem for
normal matrices \cite{bhatia}, this operator may be diagonalised and
has real, non-negative eigenvalues, which we denote
$\lambda_j^{\downarrow}(\rho_m^L)$. The arrow indicates that the
eigenvalues are arranged in decreasing order.

\begin{defn}\label{ent}
The \emph{entanglement} of the interval
$[0,L]$ relative to its complement $\Delta_m \setminus [0, L]$ is given by
\begin{equation}\label{entdef}
S(\rho^L_m) = -\tr(\rho_m^L \log_2 \rho_m^L) .
\end{equation}
\end{defn}

This quantity may be
expressed thus in terms of the eigenvalues of $\rho_m^L$:

\begin{equation}\label{diagent}
S(\rho^L_m) =
-\sum_{j=1}^{2^{L+1}} \lambda_j^{\downarrow}(\rho_m^L)\log_2
\lambda_j^{\downarrow}(\rho_m^L),
\end{equation}
where $0 \log_2 0$ is interpreted as $0$.

In Section \ref{sec:entrs}, we prove our main theorem: the order of the entanglement
scaling is at most $\log_2 L$ for the ground state in the subcritical regime. This result
follows as a corollary of the main estimate, given by Theorem
\ref{mainest}, in Section \ref{be}. In Sections \ref{sec:rcrep}--\ref{rcred}, we describe the
mapping of the density operator of the quantum Ising model to a
stochastic integral in terms of a Poisson measure (as in \cite{AKN}). The mapping
begins by considering states with $\b < \oo$ and deriving the
ground state in the limit $\b\to\oo$. This allows us to express the
matrix elements of the ground state in terms of a continuous percolation model on
a two-dimensional space--time graph, with one continuous axis
describing time. In this setting, the elements of the reduced
state are related to a random-cluster model on the same graph, but
with the addition of a `slit' along the interval $[0, L]$ at time $0$. The continuum
random-cluster model is presented in detail in Section \ref{rc}.
Section \ref{be} contains the main result, which allows us to establish the
scaling of the entanglement entropy, while in Section \ref{rwm} we explain
the technique of ratio-weak mixing on which the proof is based.
The extension of our results to disordered systems is discussed in
Section \ref{sec:disorder}.

\section{Entropy of the reduced state}\label{sec:entrs}

In this section, we study the behaviour of the entropy of the
reduced state $\rho_m^L$ in the subcritical regime (with $\theta=\lam/\de$
small). In order to derive an adequate upper bound on the entropy,
we shall analyze the influence on the spectrum of the reduced
density operator produced by imposing a change in the boundary
conditions of the spin chain. Specifically, we consider the
distance between the largest eigenvalues of two states defined on
$[0, L]$ with respect to two different lattices, $\Delta_m$,
$\Delta_n$, with $m\leq n$. The entropy will be estimated by
studying the operator norm \beq\label{sup} \| \rho^L_m- \rho^L_n
\|\equiv \sup_{\|\psi\|=1} \big| \langle \psi |\rho^L_m-
\rho^L_n|\psi\rangle\big|, \eeq where the supremum is taken over
all vectors $|\psi\rangle \in \sH_L$ with unit $L^2$-norm
belonging to the Hilbert space $\sH_L$ of spins in $[0, L]$. We
shall see in Section \ref{rcred} that $\|\rho_m^L - \rho_n^L\|$
may be expressed in terms of a certain \rc\ representation of the
quantum Ising model. In Sections \ref{be} and \ref{rwm} we shall use
a coupling of random-cluster measures and the method of `ratio
weak-mixing' to prove the following.

\begin{thm}\label{mainest2}
Let $\lam,\de\in(0,\oo)$ and write $\th = \lam/\de$.
There
exist constants $\a,C \in(0,\oo)$ depending on $\th$ only, 
and a constant $\g=\g(\th)$ satisfying
$0<\g<\oo$ if $\th<1$, such that, for all $L\ge 1$,
\begin{equation}\label{eq:rcbound}
\|\rho_m^L-\rho_{n}^L\| \le \min\{2, C L^\alpha e^{-\gamma m}\},
\qquad 2\le m\le n.
\end{equation}
Furthermore, we may find such $\g$ satisfying $\g\to\oo$ as $\th\downarrow 0$.
\end{thm}

\begin{proof}
That $\|\rho_m^L-\rho_{n}^L\| \le 2$ is a consequence of the fact that the $\rho_m^L$ are
density operators. An upper bound of the form $C' L^\alpha e^{-\gamma m}$
holds by Theorem \ref{mainest} and the preceding discussion whenever
$m \ge M$ for suitable $M=M(\th)$. Inequality \eqref{eq:rcbound}
follows on replacing $C'$ by $C=e^{\g M}\max\{C',2\}$.
\end{proof}

We shall apply \eqref{eq:rcbound} iteratively in order to obtain
an upper bound for the decay of the vector of eigenvalues
$\{\lambda_j^{\downarrow}(\rho_m^L): j=1,2,\dots\}$, valid for all large $m$.
The proof makes use of the following decomposition property, valid
for any pure state of a bipartite system, see \cite{nc}.

\begin{thm}[Schmidt decomposition]\label{thm:schmidt}
Let $|\psidm\rangle$ be the pure ground state of the composite
system $[0, L] \cup (\Delta_m\setminus [0, L])$. There exist
orthonormal bases $\{|u_j\rangle_{[0, L]},
|v_k\rangle_{\Delta_m\setminus [0, L]}\}$ for the states of $[0, L],
\Delta_m\setminus [0, L]$ respectively, such that
\begin{equation}
|\psidm\rangle = \sum_{j=1}^{s}
\sqrt{\lambda_j^{\downarrow}(\rho_m^L)}\,|u_j\rangle_{[0,
L]}|v_j\rangle_{\Delta_m\setminus [0, L]},
\end{equation}
where $s$, the \emph{Schmidt rank}, is given by $s = \min\{2^{L+1},
2^{2m}\}$.
\end{thm}

\begin{proof}
We begin by writing $|\psi_m\rangle$ in terms of an orthonormal
basis $|\alpha\rangle_{[0, L]}|\beta\rangle_{\Delta_m\setminus [0,
L]}$ where $|\alpha\rangle_{[0, L]}$ (respectively,
$|\beta\rangle_{\Delta_m\setminus [0, L]}$) is an orthonormal basis
for the spins in $[0, L]$ (respectively, $\Delta_m\setminus [0,
L]$):
$$
|\psi_m\rangle = \sum_{\alpha=1}^{2^{L+1}}\sum_{\beta=1}^{2^{2m}}
\psi_{\alpha\beta}^{[m]}\, |\alpha\rangle_{[0, L]}
|\beta\rangle_{\Delta_m\setminus [0,L]},$$ where
$$
\sum_{\alpha=1}^{2^{L+1}}\sum_{\beta=1}^{2^{2m}}|\psi_{\alpha\beta}^{[m]}|^2
= 1.
$$

The coefficients $\psi_{\alpha\beta}^{[m]}$ constitute a
$2^{L+1}\times 2^{2m}$ matrix and, as such, we can apply the
singular-value decomposition \cite{bhatia} to write
\begin{equation*}
\psi_{\alpha\beta}^{[m]} = \sum_{j = 1}^{s} U_{\alpha j} d_{j}
V_{j\beta},
\end{equation*}
where $s = \min\{2^{L+1}, 2^{2m}\}$, $U_{\alpha j}$ is a
$2^{L+1}\times s$-sized isometry, $d_j \ge 0$ for all $j = 1, 2,
\ldots, 2^{L+1}$, and $V_{j\beta}$ is an $s\times 2^{2m}$-sized
isometry. Defining
\begin{equation*}
|u_j\rangle_{[0, L]} = \sum_{\alpha = 1}^{s} U_{\alpha j}
|\alpha\rangle_{[0, L]},
\quad
|v_j\rangle_{\Delta_m\setminus [0, L]} = \sum_{\beta = 1}^{s}
V_{j\beta } |\beta\rangle_{\Delta_m\setminus [0, L]},
\end{equation*}
we see, because $U$ and $V$ are isometries, that $\{|u_j\rangle_{[0,
L]}\}$ and $\{|v_j\rangle_{\Delta_m\setminus [0, L]}\}$ are
orthonormal sets of vectors for the spins in $[0, L]$ and
$\Delta_m\setminus [0, L]$, respectively.

A simple computation shows that the reduced density operator
$\rho_m^L$ for the spins in $[0, L]$ is given by
\begin{equation}
\rho_m^L = \sum_{j=1}^{s} d_j^2 |u_j\rangle_{[0, L]}\langle u_j|
\end{equation}
and so we identify $d_j=\sqrt{\lambda_j^{\downarrow}(\rho_m^L)}$,
after re-ordering the index $j$ if necessary. Note that the rank of
$\rho_m^L$ is less than or equal to the Schmidt rank of
$|\psi_m\rangle$.
\end{proof}

We compute the entanglement of $[0, L]$ with respect to the rest
of the system as in (\ref{diagent}),
\begin{equation}
S(\rho_m^L) = -\sum_{j=1}^{s} \lambda_j^{\downarrow}(\rho_m^L)\log_2
\lambda_j^{\downarrow}(\rho_m^L).
\end{equation}
Here is our main theorem. 
With the exception of the natural logarithm function
$\ln$, all logarithms in the remainder of this section are taken to base 2.

\begin{thm}\label{entest}
Consider the quantum Ising model \eqref{ham} on $n = 2m+L+1$ spins, with
parameters $\lam$, $\de$, and let $\g$, $\a$, $C$ be as in Theorem \ref{mainest2}.
If
$\gamma > 4\ln 2 $, there exist constants $c_1$ and
$c_2$ depending on $\g$ only such that
\begin{equation}\label{eq:entropybound}
S(\rho_m^L) \le c_1 \log_2 L + c_2,\qquad
m\ge 0.
\end{equation}
\end{thm}

In summary, the entanglement entropy $S(\rho_m^L)$
is at most logarithmic in $L$ if the field strength $\de$ is
sufficiently large. The bound $4\ln 2$ is sufficient but not necessary,
and may be improved with more care in the proof. We do not know how to replace this
condition by $\g>0$.

We believe that the upper bound (\ref{eq:entropybound}) is, in many
cases, not tight. For the translation-invariant subcritical case
$\theta =\lam/\de < 2$ it is expected, on physical grounds,
that the upper bound can be improved to a
constant. (See \cite{amico} and the
references therein for an extensive
review of the physical arguments for entanglement scaling in
non-critical and critical quantum spin models.) 
Renormalisation group arguments and conformal field theory
methods suggest that, at a critical point, the upper bound should
scale with $\log L$. For $\theta>2$ the system enters the
supercritical regime where the system has two degenerate ground
states, and the ground state is no longer a pure state. Nonetheless,
it is expected that the entropy of a block is again bounded by a
constant. For higher dimensions $d\ge 2$ our argument breaks down
because the number of non-zero Schmidt coefficients for a
distinguished region grows too quickly for our perturbation argument.

The proof of Theorem~\ref{entest} follows an iterative inductive
procedure, where at each step the distance $k$ from the boundary of
$[0, L]$ is increased and the spectrum of the relative density
operator $\rho_k^L$ is estimated. We illustrate the procedure by the
following simple case: consider the ground state $|\psidz \rangle$
for the Ising model defined on only $L+1$ spins. In this case the
reduced density operator $\rho_0^L$ for $[0, L]$ is exactly
$\rho_0^L = |\psidz \rangle\langle \psidz|$, i.e., a pure state,
with entropy $S(\rho_0^L) = 0$. When $m=1$, the reduced density
operator $\rho_1^L$ for the region $[0,L]$ is mixed, but it has at
most $2^2$ non-zero eigenvalues. This follows from the Schmidt
decomposition applied to the ground state $|\psi_1\rangle$ across
the bipartition $[0,L]\cup(\Delta_1\setminus [0, L])$. Thus, the
entropy of the block $[0, L]$ is bounded above by $S(\rho_1^L) \le
2$. Consider now the reduced density operator $\rho_k^L$. By the Schmidt
decomposition, the operator $\rho_k^L$ has at most $2^{2k}$ non-zero
eigenvalues. Assume that $2k < L+1$, and consider the addition of
a single spin at either boundary. The new reduced density operator $\rho_{k+1}^L$
has at most four times
as many non-zero eigenvalues as $\rho_{k}^L$. However, by \eqref{eq:rcbound},
\begin{equation}
\|\rho_k^L-\rho_{k+1}^L\| \le \min\{2, C L^\alpha e^{-\gamma k}\},
\end{equation}
so that the eigenvalues of $\rho_k^L$ remain close to those of
$\rho_{k+1}^L$.

\begin{proof}[Proof of Theorem~\ref{entest}]
Let $K=\lceil \gamma^{-1}\ln (CL^\alpha)\rceil$, with $C$, $\a$, $\g$
as in Theorem \ref{mainest2}. We shall assume that 
$m, K \ge 2$, $\g > 4\ln 2$. There are two cases, depending on
whether $m\le K$ or $m > K$. Assume first that $2\le m\le K$. The 
rank of $\rho_m^L$ equals the Schmidt rank $2^{2m}$ of $|\psi_m\rangle$.
Therefore,
$$
S(\rho_m^L) \le \sup_\rho \left\{-\sum_{j=1}^{2^{2m}} \rho_j \log \rho_j\right\},
$$
where the supremum is over all non-negative sequences 
$\rho = (\rho_j: 1\le j\le 2^{2m})$
with sum 1. Hence,
\be\label{eq:1002}
S(\rho_m^L) \le \log 2^{2m} = 2m \le 2K, \qquad m \le K.
\ee

Assume next that $m \ge K$. We shall apply the following theorem, 
see \cite{bhatia}.

\begin{thm}[Weyl perturbation theorem]\label{thm:weyl}
For Hermitian operators $A$ and $B$ on a Hilbert space
of dimension $n$,
\begin{equation}
\max_j \bigl|\lambda_j^{\downarrow}(A)-\lambda_j^{\downarrow}(B)\bigr| \le
\|A-B\|.
\end{equation}
\end{thm}

Let $\eps(r)= C L^\alpha e^{-\gamma (K+r)}$, and note by the definition
of $K$ that
\be\label{eq:eps}
\eps(r) \le e^{-\g r},\qquad r \ge 0.
\ee
Setting $A=\rho_K^L$, $B=\rho_{K+1}^L$ in Theorem \ref{thm:weyl}, we deduce by
\eqref{eq:rcbound} that
\begin{equation}
\max_j \bigl|\lambda_j^{\downarrow}(\rho_K^L)-\lambda_j^{\downarrow}(\rho_{K+1}^L)\bigr|
\le \eps(0).
\end{equation}
Therefore,
\begin{alignat}{2}\label{eq:rhokbound}
|\lambda_j^{\downarrow}(\rho_{K+1}^L)| &\le
\lambda_j^{\downarrow}(\rho_{K}^L)+ \epsilon(0), \qquad
&&j=1,2, \dots, 2^{2K}, \nonumber\\
|\lambda_j^{\downarrow}(\rho_{K+1}^L)| &\le \epsilon(0),
&&j=2^{2K}+1,2^{2K}+2, \dots, 2^{2(K+1)}.
\end{alignat}

We shall now iterate this process in order to obtain a bound on the eigenvalues of
$\rho_{K+r}^L$, for $r\ge 1$.
There are three cases:

\begin{romlist}
\item $j\le 2^{2K}$, in which case
\begin{equation}
\lambda_j^{\downarrow}(\rho_{K+r}^L) \le
\lambda_j^{\downarrow}(\rho_{K}^L) + \sum_{l=0}^{r-1} \epsilon(l);
\end{equation}

\item $2^{2K} \le 2^{2(K+s)} < j \le 2^{2(K+s+1)} \le 2^{2(K+r)}$, in
which case
\begin{equation}
\lambda_j^{\downarrow}(\rho_{K+r}^L) \le \sum_{l=s}^{r-1}
\epsilon(l);
\end{equation}

\item $2^{2(K+r)} < j$, in which case
\begin{equation}
\lambda_j^{\downarrow}(\rho_{K+r}^L) = 0.
\end{equation}
\end{romlist}

Let $s = \lfloor\frac{1}{2}\log j\rfloor - K$, so that, by \eqref{eq:eps},
\begin{alignat}{2}
\lambda_j^{\downarrow}(\rho_{m}^L) &\le
\lambda_j^{\downarrow}(\rho_{K}^L) +  \sum_{l=0}^{\infty}
e^{-\gamma l}, \qquad &&j \le 2^{2K},\nonumber\\
\lambda_j^{\downarrow}(\rho_{m}^L) &\le  \sum_{l=s}^{\infty}
e^{-\gamma l}, &&2^{2K} < j,\nonumber
\end{alignat}
which is to say that
\begin{alignat}{2}
\lambda_j^{\downarrow}(\rho_{m}^L) &\le
\lambda_j^{\downarrow}(\rho_{K}^L) + c_0
, \qquad &&j \le 2^{2K},\nonumber\\
\lambda_j^{\downarrow}(\rho_{m}^L) &\le c_0e^{-\gamma s},
&&2^{2K} < j,\label{eq:1001}
\end{alignat}
where 
\be\label{eq:1000}
c_0= \frac 1{1-e^{-\gamma}} \le \frac43.
\ee
By \eqref{eq:1001},
\be\label{eq:eigbounds}
\lambda_j^{\downarrow}(\rho_{m}^L) \le c_0' j^{-\xi},
\qquad 2^{2K} < j,
\ee
where $\xi = {\gamma}/(2\ln 2) > 2$ and $c_0'=c_0'(L)=c_0 e^{\gamma(K+1)}$.

By \eqref{diagent},
\be\label{eq:1003}
S(\rho_m^L) =  S_1 + S_2,
\ee
where 
$$
S_1 = -\sum_{j=1}^{\nu} \lambda_j^{\downarrow}(\rho_m^L)
\log \lambda_j^{\downarrow}(\rho_m^L), \quad
S_2 =
-\sum_{j=\nu+1}^{2^{L+1}} \lambda_j^{\downarrow}(\rho_m^L)
\log \lambda_j^{\downarrow}(\rho_m^L),
$$ 
where $\nu$ ($\ge 2^{2(K+2)}$) is an integer to be chosen later.
We shall bound $S_1$ and $S_2$ separately. 
Since the $\lambda_j^{\downarrow}(\rho_m^L)$, $1\le j\le \nu$, are non-negative 
with sum $Q$ satisfying $Q\le 1$, 
\begin{equation}\label{eq:103}
S_1 \le \log \nu.
\end{equation}

We shall use the tail estimate \eqref{eq:eigbounds}
to bound $S_2$, making use of the fact that the function $f(x) =
-x\log x$ satisfies:  $f(0) =
0$, and $f(x) < f(y)$ whenever $0<x<y< e^{-1}$. 

By \eqref{eq:1000}, \eqref{eq:eigbounds}, and the definition of $\xi$,
$$
\lambda_j^\downarrow (\rho_m^L)
\le \frac{c_0'}{ j^{\xi}} < e^{-1}, \qquad j \ge 2^{2(K +2)},
$$ 
and so, recalling that $\nu \ge 2^{2(K+2)}$ and $\xi>2$,
\begin{equation*}
\begin{split}
S_2 &\le -\sum_{j=\nu+1}^{2^{L+1}}\frac{ c_0'}
{j^{\xi}}\log \left(\frac{c_0'}{j^{\xi}}\right) 
\le -\sum_{j=\nu+1}^{\infty} \frac{c_0'}{ j^{\xi}}\log \left(\frac{c_0'}{j^{\xi}}\right) \\
&\le -[c_0'\log c_0'] \sum_{j=\nu+1}^\infty \frac{1}{j^{\xi}} +
\frac{\xi c_0'}{\ln 2}\sum_{j=\nu+1}^{\infty} \frac{\ln j}{j^\xi} \\
&\le |c_0'\log c_0'| \int_{\nu}^\infty \frac1{x^{\xi}} \,dx +
\frac{\xi c_0'}{\ln 2}\int_{\nu}^\infty \frac1{x^{\xi}}\ln x \,dx \\
&\le \frac{c_0'\nu^{1-\xi}}{\xi-1}\left(|\log c_0'| +
\xi \log \nu + \frac{\xi}{\xi-1}\right). \\
\end{split}
\end{equation*}
We now set $\nu = \lceil e^{\gamma (K + 1)}\rceil$
to obtain
\begin{equation}\label{eq:s2bound}
S_2 \le  c_1 K + c_2,
\end{equation}
for suitable constants $c_1$, $c_2$ depending on $\g$ only. 
By \eqref{eq:1003}--\eqref{eq:s2bound},
\be\label{eq:1004}
S(\rho_m^L) \le c_1' K + c_2', \qquad m \ge K,
\ee
which may be combined with \eqref{eq:1002} to obtain \eqref{eq:entropybound}
with adjusted constants.
\end{proof}

\section{Percolation representation of the ground state}\label{sec:rcrep}

Aizenman, Klein, and Newman \cite{AKN} derived a random-cluster
representation for the thermal state of the quantum Ising
Hamiltonian \eqref{ham}, thereby relating spin-correlation
properties to graph-connectivity properties. In this representation,
the thermal density operator, defined as
\begin{equation}\label{eq:trace}
\rho_m(\beta) = \frac{e^{-\beta H_{m}}}{\tr(e^{-\beta H_{m}})},
\qquad \b = T^{-1}>0,
\end{equation}
is described by a
stochastic integral with respect to a Poisson measure. This Poisson measure is
defined on the space--time graph $\Lambda_{m,
\beta}=\Delta_m\times[0, \beta]$, generated by associating a
continuous (imaginary) time variable $t\in [0, \beta]$ to each
site $x\in \Delta_m$. We refer to a line of the form $\{x\}\times[0,\b]$
as the \emph{time-line} at the site $x$.

For completeness, we reproduce here the derivation of the
random-cluster representation of the ground state, and we derive
the corresponding representation for the reduced state on $[0,
L]$. Note that the derivations are valid with the line $\Delta_m$
replaced by any finite graph $G$. By \eqref{ham} with $\nu =
\frac12 \sum_{\langle x, y\rangle} \lambda\II$ and $\II$ the
identity operator,
\begin{equation}\label{expon}
e^{-\beta (H_m+\nu)}= e^{-\beta(U+V)},
\end{equation}
where
$$
U=-\delta \sum_{x} \sigma^{(1)}_x,\quad
V=- \frac12
\sum_{\langle x, y\rangle} \lambda (\sigma^{(3)}_x
\sigma^{(3)}_y - \II),
$$
and the second summation is
over all neighbouring pairs in $\De_m$. Although these two terms do not
commute, we may use the so-called Lie--Trotter formula (see, for
example, \cite{schulman}) to factorize the exponential in \eqref{ham} into
\emph{single-site} and \emph{two-site} contributions due to $U$
and $V$, respectively. By the Lie--Trotter formula,
$$
e^{-(U+V)\Delta t} =
e^{-U \Delta t} e^{-V \Delta t} + \O(\Delta t^2).
$$
We divide the interval $[0, \beta]$ into $N$
parts each of length $\Delta t=1/N$, and deduce that
\begin{equation}
e^{-\beta(U+V)} = \lim_{\Delta t \rightarrow 0} \left(e^{-U \Delta
t} e^{-V \Delta t}\right)^{\beta/\Delta t}.
\end{equation}
We then expand the exponential, neglecting all terms of order
$\o(\Delta t)$, to obtain
\begin{multline}
e^{-\beta (H_m+\nu)} =\\
\lim_{\Delta t \rightarrow 0} \left( \prod_{x} \left[(1- \delta
\Delta t) \II+ \delta \Delta t P^1_x\right] \prod_{\langle
x, y\rangle}\left[(1-\lambda \Delta t) \II+ \lambda \Delta
t P^{3}_{x,y} \right]\right)^{\beta/\Delta t},\label{eq:dec}
\end{multline}
where $P^1_x = \sigma^{1}_{(x)} + \II$ and $P^{3}_{x,y} =
\frac12(\sigma^{(3)}_x \sigma^{(3)}_y+ \II)$.

Let $B$ be the set of basis vectors $|\s\rangle$ for $\sH$ of the form $|\sigma\rangle =
\bigotimes_x|\pm\rangle_x$. There is a natural one--one correspondence between $B$
and the space $P=\prod_{x=-m}^{m+L}\{-1,+1\}$.
We shall sometimes speak of members of $P$ as basis vectors,
and of $\sH$ as the Hilbert space
generated by $P$. Similarly, the space $\sH_L$ of spins indexed by the interval
$[0,L]$ may be viewed as being generated by $P_L=\prod_{x=0}^L\{-1,+1\}$.

The stochastic-integral representation may be obtained from
(\ref{eq:dec}) by inserting the resolution of the identity
\be\label{eq:resolution}
\sum_{\sigma\in P} |\sigma\rangle\langle \sigma|=\II
\ee
between any two factors of the products.
The product \eqref{eq:dec} contains a collection of operators
acting on sites $x$ and on neighbouring pairs $\langle x,y\rangle$.
By labelling the
time-segments as $\Delta t_1,\Delta t_2, \dots, \Delta t_N$ in $[0, \beta]$,
and neglecting terms of order $\o(\Dt)$, we
may see that each given time-segment arising in \eqref{eq:dec}
contains
one of: the identity $\II$; an operator of the
form $P_x^{1}$; an operator of the form $P_{x,y}^{3}$. Each
such operator occurs in the time-segment with a certain weight.

Let us consider the action of these operators on the
states $|\sigma\rangle$ for each infinitesimal time interval
$\Delta t_i$, $i \in \{1, 2,\dots, N\}$. The matrix elements of each
of the single-site operators are given by
\begin{equation}\label{eq:death}
\langle \sigma '| \sigma^{(1)}_x+\II| \sigma\rangle=
\delta_{\sigma_x', \sigma_x}+ \delta_{\sigma_x',
\overline{\sigma}_x} =1,
\end{equation}
where $\sigma_x$ is the value of the spin at $x$ in the (product)
basis vector $|\sigma\rangle$, and $\overline{\sigma}_x$ is the
opposite spin to $\sigma_x$. When it occurs in some time-segment
$\Delta t_i$, we place a mark in the interval $\{x\}\times \Delta
t_i$, and we call this mark a \emph{death}. Such a death has a
corresponding weight $\delta \Dt + \o(\Dt)$.

The matrix elements involving neighbouring pairs $\langle x,y\rangle$ yield
\begin{equation}\label{eq:bridge}
\tfrac12\langle \sigma'_x \sigma'_y| \sigma^{(3)}_x
\sigma^{(3)}_y+ \II |\sigma_x \sigma_y\rangle =
\delta_{\sigma_x,\sigma'_{x}}\delta_{\sigma_y,\sigma'_{y}}
\delta_{\sigma_x,\sigma_y}.
\end{equation}
When this occurs in some time-segment $\Delta t_i$, we place a
connection, called a \emph{bridge}, between the intervals
$\{x\}\times \De t_i$ and $\{y\}\times \De t_i$. Such a bridge has
a corresponding weight $\lambda \Dt +\o(\Dt)$.

In the limit $\Dt \rightarrow 0$, the spin operators
generate thus a Poisson process with intensity $\de$ of deaths in each
time-line $\{x\}\times[0, \beta]$,
and a Poisson process with intensity $\lam$ of bridges between each
pair $\{x\}\times [0,\b]$, $\{y\}\times [0,\b]$ of time-lines,
for neighbouring $x$ and $y$. This is an independent family
of Poisson processes. We write $D_x$ for the set of deaths at the site $x$,
and $B_{x,y}$ for the set of bridges between neighbouring sites $x$ and $y$.
The configuration space is the set $\Om_{m,\b}$ containing all finite sets of
deaths and bridges, and we may assume without loss of generality that
no death is the endpoint of any bridge.

For two point $(x,s), (y,t) \in \La_{m,\b}$, we write
$(x,s) \lra (y,t)$ if there exists a path from the first to the
second that traverses time-lines and bridges but crosses no death.
A \emph{cluster} is a maximal subset $C$ of $\La_{m,\b}$ such that
$(x,s) \lra (y,t)$ for all $(x,s), (y,t) \in C$. Thus the connection relation
$\lra$ generates a percolation process on $\La = \La_{m,\b}$, and we write
$\PLlb$ for the probability measure corresponding
to the weight function on the configuration
space $\Om_{m,\b}$. That is, $\PLlb$ is the measure governing a family
of independent Poisson processes of deaths (with intensity $\de$)
and of bridges (with intensity $\lam$). The ensuing percolation
process has been studied in \cite{BG}.

We shall later need to count the number of clusters of a
configuration $\om\in\Om_{m,\b}$ subject to any of four possible
boundary conditions, of which we specify two next (the other two
appear in the next section). The meaning of
\emph{periodic boundary condition}
is that any clusters containing two points of
the form $(x,0)$ and $(x,\b)$, for some $x\in\De_m$,
are deemed to be the same cluster, and they
contribute only 1 to the total cluster count. The meaning of
\emph{wired boundary condition}
is that any clusters containing two points of
the form $(x,0)$ and $(y,\b)$, for $x,y\in\De_m$,
are deemed to be the same cluster and
contribute only 1 to the total count.  We write $\kp(\om)$
(\resp, $\kw(\om)$) for the number of clusters of $\om$
subject to the periodic (\resp, wired) boundary condition.
Note that $\kw(\om)-1$ is the number of clusters of $\om$
(with free boundary conditions) that do not intersect
$[-m,m+L]\times \{0,\b\}$.

Equations \eqref{eq:death}--\eqref{eq:bridge} are
to be interpreted as saying the following. In calculating the
operator $e^{-\b (H_m+\nu)}$, one averages over contributions from
realizations of the Poisson processes, on the basis
that the quantum spins are constant on every cluster of the
corresponding percolation process, and each such spin-function
is equiprobable.

More explicitly,
\begin{equation}\label{eq:int}
e^{- \beta (H_m+\nu)}= \int d\PLlb(\omega)
\left(\sT\prod_{(x,t)\in D}\,\prod_{(\langle
x,y\rangle,t')\in B} P_x^1(t) P_{x,y}^{3}(t')\right),
\end{equation}
where $\sT$ denotes the time-ordering of the terms in
the products, and $B$ (\resp, $D$) is the set of all bridges (\resp, deaths)
of the configuration $\om\in\Om_{m,\b}$. The $P_x^1(t)$ and $P_{x,y}^3(t)$
are to be interpreted as the relevant operators encountered
at the deaths and bridges of $\om$.

Let $\om\in\Om_{m,\b}$. Let $\Si(\om)=\Sigma_{m,L}(\om)$ be the
space of all functions $s: \Delta_m\times[0,
\beta] \rightarrow \{-1,+1\}$ that are constant on the clusters of $\om$,
and let $\mu_\om$ be the
counting measure on $\Si(\om)$. Let $K(\omega)$ be the time-ordered
product of operators in \eqref{eq:int}. We may evaluate the matrix
elements of $K(\omega)$ by inserting the resolution of the
identity between any two factors in the product, obtaining thus that
\begin{equation}\label{eq:k}
\langle\sigma'|K(\omega)|\sigma \rangle = \sum_{s\in\Si(\om)}
1\{s(\cdot,0)=\sigma \}  1\{s(\cdot,\beta)=\sigma'\},
\qquad \s,\s'\in P,
\end{equation}
where $1\{A\}$, and later $1_A$, denotes the indicator function of $A$.
This is the number of spin-allocations to the clusters of $\om$
with given spin-vectors at times 0 and $\b$.

The matrix elements of the
density operator $\rho_m(\b)$ are therefore given by
\begin{equation}\label{eq:thing0}
\langle \sigma'|\rho_m(\b)|\sigma \rangle =
\frac1{Z_m}\int
1\{s(\cdot, 0)=\sigma \} 1\{s(\cdot,\beta)=\sigma'\}\, d\mu_\om(s)\,d\PLlb(\om),
\end{equation}
for $\s,\s'\in P$,
where
\be\label{eq:pf}
Z_m = Z_m(\b) = \tr(e^{-\b (H_m+\nu)})
\ee
is the \emph{partition function}.
Thus,
\begin{align}
\langle \sigma'|\rho_m(\b)|\sigma \rangle &= \frac{1}{Z_m}\int
d\PLlb(\omega) \sum_{s\in\Si(\om)} 1\{s(\cdot, 0)=\sigma \} 1\{\sigma(\cdot,\beta)=\sigma'\}
\nonumber\\
&= \frac{1}{Z_m}\int
d\PLlb(\om)\, 2^{\kw(\om)-1} 1_{E(\s,\s')}(\om),
\qquad \s,\s'\in P,
\label{eq:thing}
\end{align}
where the final term in the integrand is the indicator function
of the event $E(\s,\s')$ containing all $\om\in\Om_{m,\b}$ such that:
for all $x,y\in[-m,m+L]$:
\begin{alignat}{2}
(x,0) \nlra (y,0)\quad&\mbox{whenever}\quad  &\s_x\ne \s_y,\nonumber\\
(x,\b)\nlra (y,\b)\quad &\mbox{whenever}  &\s_x'\ne \s_y',\nonumber\\
(x,0)\nlra (y,\b)\quad &\mbox{whenever}  &\s_x \ne \s_y'.\nonumber
\end{alignat}
See Figure~\ref{fig:randomcluster} for an illustration of the
space--time configurations contributing to the Poisson integral
(\ref{eq:thing}) for the matrix elements of $\rho_m(\beta)$.

\begin{figure}[ht]
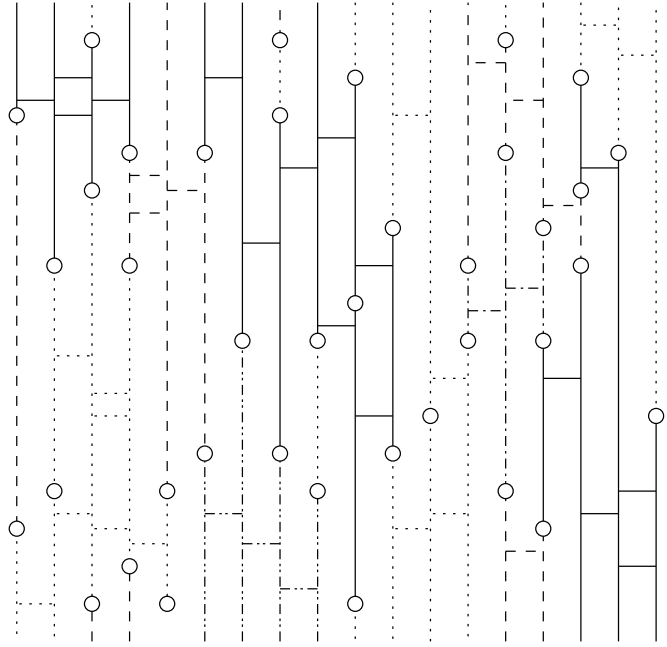

\begin{center}
\input gosfig3.pstex_t
\caption{An example of a space--time
configuration contributing to the Poisson integral (\ref{eq:thing}).
The cuts are shown as circles and the distinct connected clusters
(each of which contributes a factor $2$ to the term $2^{\kw(\om)}$) are
indicated with different line-types.}\label{fig:randomcluster}
\end{center}
\end{figure}

On setting $\s=\s'$ in \eqref{eq:thing} and summing over $\s\in P$, we find that
\begin{equation}
Z_m = \tr(e^{- \beta (H_m+\nu)}) = \int
2^{\kp(\omega)}\, d\PLlb(\omega).
\end{equation}

A method was developed in \cite{AKN} (as amplified in the next
section) to represent $\langle \sigma'|\rho_m(\b)|\sigma \rangle $
as a certain probability, and to prove that it converges as
$\b\to\oo$. In particular, it was shown in \cite{AKN} that the
ground state $\rho_m = |\psidm\rangle \langle \psidm|$ satisfies
\begin{equation}\label{eq:last}
\rho_m = \lim_{\beta\rightarrow\infty} \frac1{Z_m} e^{-\beta (H_m+\nu)}.
\end{equation}

\section{Percolation representation of the reduced
state}\label{rcred}

The analysis of the last section may be repeated for the
reduced density operator $\rho^L_m(\b)$ by tracing \eqref{eq:int} over a complete set
of states of the spins indexed by $\Delta_m \setminus [0, L]$. The corresponding
boundary condition for the configuration $\om\in\Om_{m,\b}$
turns out to be \emph{partially periodic}, in that any two clusters
of $\om$ containing points of the form $(x,0)$ and $(x,\b)$,
for some $x\in [-m,-1] \cup [m+1,m+L]$, are deemed to be the same
cluster and contribute only 1 to the total cluster count. No such
assumption is made for sites $x\in[0,L]$, and we refer to the boundary
condition on $[0,L]$ as \emph{free}. Let $\kpp(\om)$
be the number of clusters of $\om$ subject to the partially periodic
boundary condition. We shall need a fourth way to count clusters also, as follows.
The \emph{periodic/wired} boundary condition is that derived from
the partially periodic condition by the additional assumption
of a wired condition on $[0,L]$: any two clusters
of $\om$ containing points of the form $(x,0)$ and $(y,\b)$,
for some $x,y\in[0,L]$, are deemed to be the same
cluster and contribute only 1 to the total cluster count.
We write $\kpw(\om)$ for the number
of clusters with the periodic/wired boundary condition.
Note that $\kpw(\om)-1$ is the number
of clusters of $\om$ (with the partially periodic
boundary condition) that do not intersect $[0,L]\times \{0,\b\}$.

\begin{figure}[ht]
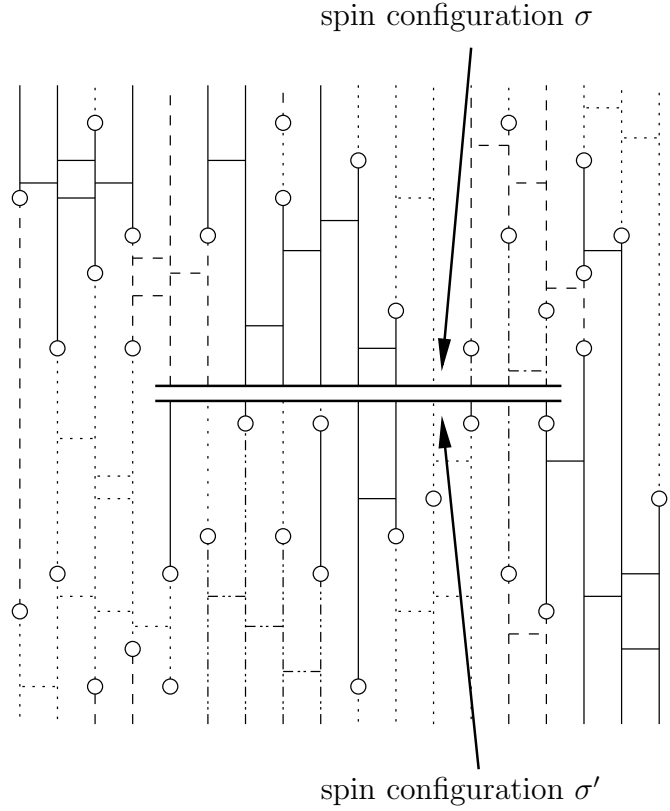

\begin{center}
\input gosfig4.pstex_t
\caption{An example of a space--time configuration
contributing to the matrix elements for the reduced density operator
$\rho_m^L(\beta)$. The box has partially periodic boundary conditions
and is drawn in such a way that the slit $S_L$ is at the
centre. The spin configurations on the top and the bottom
of the cut, and the connected clusters for this new cut geometry are
indicated.}\label{fig:rdo}
\end{center}
\end{figure}

As in \eqref{eq:thing0}--\eqref{eq:thing}, the matrix
elements of the reduced state $\rho_m^L(\b)$ are given by
\begin{equation}\label{redstate}
\langle \sigma_L'|\rho_m^L(\b)|\sigma_L
\rangle =
\frac{1}{Z_m}\int d\PLlb(\om)\,
2^{\kpw(\om)-1} 1_{E(\s_L,\s_L')}(\om),
\quad \s_L,\s_L'\in P_L,
\end{equation}
where $E(\s_L,\s_L')$ is the event that: if $x,y\in [0,L]$ are such
that $\s_{L,x} \ne \s_{L,y}'$ then $(x,0) \nlra^\pp (y,\b)$. Here,
$\lra^\pp$ denotes the connectivity relation subject to the
partially periodic boundary condition. See Figure~\ref{fig:rdo} for
an illustration of the slit space--time, and of the connected
clusters contributing to the matrix elements of $\rho_m^L(\beta)$.

We shall study the entropy of the reduced state via the operator
norm of \eqref{sup}. Let $|\psi\rangle \in \sH_L$ have unit $L^2$-norm,
so that
$$
|\psi\rangle = \sum_{\s_L\in P_L}c(\s_L) \s_L
$$
for some function $c:P_L \to \CC$ with
$\sum_{\s_L\in P_L} c(\s_L) \ol{c(\s_L)} = 1$.
Then
\be\label{eq:0}
\langle \psi| \rho_m^L(\b) |\psi\rangle =
\frac 1{\amb} \sum_{\s_L,\s_L'\in P_L} c(\s_L)\ol{c(\s_L')} \phmb(\s_L, \s_L')
\ee
where
\begin{align}
\phmb(\s_L, \s_L') &= \frac 1{N_m} \langle \s_L'| \rho_m^L(\b) |\s_L\rangle,
\qquad \s_L,\s_L'\in P_L,\label{eq:1}\\
N_m&= \sum_{\s_L,\s_L'\in P_L} \langle \s_L'| \rho_m^L(\b) |\s_L\rangle,
\label{eq:2}
\end{align}
and
\be\label{eq:amdef}
\amb=Z_m/N_m.
\ee
We shall see in the next sections
that \eqref{eq:0}--\eqref{eq:2} may be written in terms of a certain probability
measure on $\Om_{m,\b}$ called the \emph{\rc\ measure}.

\section{The continuum \rc\ model}\label{rc}
Perhaps the best way to express the percolation representations of the 
ground and reduced states is in terms of the so-called \rc\ model
on $\ZR$. We summarise the definition and basic properties of this
model in this section, using the language of probability theory.
The remaining part of the paper is a self-contained account of the
model, and includes the proof of Theorem \ref{mainest2}, see Theorem \ref{mainest}.
Of special interest will be the property of
so-called ratio weak-mixing, studied earlier for the lattice case in \cite{Al1,Al2}.

We shall consider the (two-dimensional) continuum \rc\ model on the `space--time'
subset $\ZR$ of the plane. The underlying space is $\{(x,t): x\in
\ZZ,\ t\in \RR\}$, and we refer to $\ZZ$ as the space-line and
$\RR$ as the time-line. Everything proved here has a counterpart,
subject to minor changes, in the more general setting of
$\ZZ^d\times\RR$ with $d \ge 2$, but we shall restrict ourselves
to the case $d=1$.

We shall construct a family of probabilistic models on $\ZR$.
Let $\lam,\de \in (0,\oo)$. In the simplest such model, we construct
`deaths' and `bridges' as follows. For each $x\in\ZZ$, we select a Poisson
process $D_x$ of points in $\{x\}\times\RR$ with intensity $\de$; the
processes $\{D_x: x\in \ZZ\}$ are independent, and the points in the $D_x$ are termed
`deaths'. For each $x\in\ZZ$, we select a Poisson process $B_x$ of points
in $\{x+\frac12\}\times\RR$ with intensity $\lam$;
the processes $\{B_x: x\in\ZZ\}$ are independent of each other and of the
$D_y$. For each $x\in\ZZ$ and each $(x+\frac12,t)\in B_x$, we
draw a unit line-segment in $\RR^2$ with endpoints $(x,t)$
and $(x+1,t)$, and we refer to this as a `bridge' joining
its two endpoints.  For $(x,s), (y,t) \in \ZZ\times\RR$, we write
$(x,s)\lra (y,t)$ if there exists a path $\pi$ in $\RR^2$ with
endpoints $(x,s)$, $(y,t)$ such that: $\pi$ comprises sub-intervals
of $\ZR$ containing no deaths, together possibly with bridges.
For $\La,\De \subseteq \ZR$, we write $\La\lra \De$ if there exist
$a\in \La$ and $b\in \De$ such that $a\lra b$.

For $(x,s)\in \ZR$, let $C_{x,s}$ be the set of all points
$(y,t)$ such that $(x,s)\lra(y,t)$. The clusters $C_{x,s}$ have been studied in
\cite{BG}, where it was shown in particular that
\be
\Plb(|C_0|<\oo) \,\begin{cases}
=1 &\text{ if } \theta\le 1,\\
<1 &\text{ if } \theta > 1,
\end{cases}
\ee where $0=(0,0)$ is the origin of $\ZR$, $\theta=\lam/\de$, and
$|C|$ denotes the (one-dimensional) Lebesgue measure of the cluster $C$. The process
thus constructed is a continuum percolation model in two
dimensions. As noted in \cite{BG}, it differs from the contact
model on $\ZZ$ only in that two points may be joined in the
direction of either increasing or decreasing time. See
\cite{Lig1,Lig2} for details of the contact model.

Just as the percolation model on a lattice may be generalised to the
so-called \rc\ model (see \cite{G-RC}), so may the continuum percolation model
be extended to a continuum \rc\ model. We shall work here
mostly on a bounded box rather than the whole space $\ZR$. 
Let $a,b\in\ZZ$, $s,t\in\RR$ satisfy $a \le b$, $s \le t$, 
and write $\La=[a,b]\times[s,t]$ for the box
$\{a,a+1,\dots,b\}\times [s,t]$ of $\ZR$. Its boundary $\pd\La$
is the set of all points $(x,y)\in\La$
such that: either $x\in\{a,b\}$, or $y\in \{s,t\}$, or both.
As sample space we take the
set $\Om_\La$ comprising all finite subsets (of $\La$) of deaths and bridges,
and we assume that no death is the endpoint of any bridge.
For $\om\in\Om_\La$, we write $B(\om)$ and $D(\om)$ for the sets
of bridges and deaths, respectively, of $\om$. We take as $\s$-field 
$\sF_\La$ that generated by the open sets in the associated Skorohod topology,
see \cite{BG,EK}. 

In order to maintain the link to the quantum Ising model,
we choose to impose a top/bottom periodic boundary condition on $\La$;
that is, for every $x\in [a,b]$, we identify the two points $(x,s)$
and $(x,t)$.  The remaining boundary of $\La$,
denoted $\pdh\La$, is the set of all
points of the form $(x,u)\in \La$ with $x\in\{a,b\}$. The theory developed
here is valid for more general boundary conditions.
 
Let $\PLlb$ denote the probability
measure associated with the above continuum percolation model on $\La$.
For a given configuration $\om$ of deaths and bridges on $\La$, let $k(\om)$
be the number of its clusters (subject to the top/bottom
periodic boundary condition). Let $q\in(0,\oo)$, and define the
`continuum \rc' probability measure
$\PLlbq$ by
\be\label{rcPo}
d\PLlbq(\om) = \frac 1Z q^{k(\om)}d\PLlb(\om),
\qq \om\in\Om_\La,
\ee
for an appropriate `partition function' $Z$.

The theory of the continuum \rc\ model may be developed in very much
the same way as that for the \rc\ model on a discrete lattice, see \cite{G-RC}.
We shall assume the basic theory
without labouring the calculations necessary for full rigorous proof. The
details may be obtained by following minor variants of 
the relevant strategy for the discrete case.

If $\mu$ is a probability measure and $f$ a function on some measurable space,
we denote by $\mu(f)$ the expectation of $f$ under $\mu$.

The space $\Om_\La$ is a partially ordered space with 
order relation given by: $\om_1\le\om_2$ if $B(\om_1)\subseteq B(\om_2)$ and
$D(\om_1) \supseteq D(\om_2)$. A random variable $X:\Om_\La\to\RR$
is said to be \emph{increasing}
if $X(\om)\le X(\om')$ whenever $\om\le\om'$. An event $A\in\sF_\La$ is
said to be \emph{increasing} if its indicator function $1_A$ is increasing.
Given two probability measures $\mu_1$, $\mu_2$ on
the measurable pair $(\Om_\La,\sF_\La)$, we write
$\mu_1\lest \mu_2$ if $\mu_1(X) \le \mu_2(X)$
for all bounded increasing continuous random variables $X:\Om_\La\to\RR$.

The measures $\PLlbq$ have certain
properties of stochastic ordering as the parameters $\La$, $\lam$,
$\de$, $q$ vary. There are two approaches to such stochastic inequalities,
either by working on discrete graphs and passing to a spatial limit to
obtain the continuum measures, or by working directly in the continuum.
We shall not pursue this here, but refer the reader to
\cite{BG} for a discussion of the case $q=1$.
The following two facts will be useful later. First,
$\PLlbq$ satisfies a positive-association (FKG) inequality when $q \ge 1$,
and secondly,
\be\label{stochcomp}
\PLlbq \lest \PLlb\qq \text{when } q \ge 1.
\ee

In the current paper we shall work mostly with finite-volume
measures, that is, with measures defined on boxes
of the form of $\La=[a,b]\times[s,t]$.  We assume henceforth that $q \ge 1$.
Having established the necessary estimates on
such boxes, we will pass to the \emph{vertical}
infinite-volume limit as $s\to-\oo$, $t\to\oo$.
The existence of such a limit is not explored in 
detail here, but we note the following (see \cite{AKN}).
If we work on $\La$ with top/bottom \emph{wired} or \emph{free}
boundary conditions, then the limit measures exist as
a consequence of positive 
association (very much as in the lattice case, see \cite{G-RC}).
Furthermore, the  weak limit with top/bottom periodic boundary conditions
exists and agrees with the first two limit measures whenever the latter
are equal. A sufficient condition for this is that the wired limit measure
does not percolate. Since the limit of $\La$ as $t-s\to\oo$ is
a strip of bounded width, this condition is satisfied for all
$\lam, \de\in(0,\oo)$, and therefore the limit measures exist
and do not depend on the choice of boundary condition.

The situation is slightly less clear in the doubly-infinite-volume limit,
as $\La\uparrow\ZR$. The self-dual point for the continuum \rc\ measure on
$\ZR$ is given by $\lam/\de=q$, and thus one expects the free and wired
limit measures to be equal at least whenever $\lam/\de \ne q$. 
It may be shown using duality that there is no percolation when $\lam/\de < q$,
and it follows that the weak limits 
$$
\Plbq=\lim_{\La\uparrow\ZR} \PLlbq,
\qquad q \ge 1,
$$
exist if $\lam/\de< q$. We shall make no reference to this later.

Just as the $q$-state Potts model may be coupled with a \rc\ model
on a given graph, so may we consider a continuum Potts model on 
a box $\La=[a,b]\times[s,t]$.
Let $q\in\{2,3,\dots\}$. We sample $\om$ according to $\PLlbq$, and we allocate
a randomly chosen spin from the set $\{1,2,\dots,q\}$ to each cluster
of $\om$; the points of each cluster receive a given 
spin-state chosen uniformly
at random from the $q$ possible local states, 
and different clusters receive independent spin-states.
We call the ensuing spin-configuration a $q$-state {\em continuum Potts model},
and a {\em continuum Ising model} when $q=2$. When $q=2$,
by convention we take the local spin-space to be $\{-1,+1\}$
rather than $\{1,2\}$, and this is the case of interest in
the current paper.  We note in passing that the $q$-state continuum
\rc\ model corresponds to a certain $q$-state quantum Potts model 
constructed in a manner similar to that of the quantum Ising model.

The set of spin-configurations of the continuum $q$-state
Potts model is the space $\Si_\La$ given as follows. Let
$\sF$ be the set of finite subsets of $\La$. For $D\in\sF$, let $J(D)$ be the
set of maximal intervals of the time-lines that contain no point
in $D$ (subject to the top/bottom boundary condition on $\La$). 
The space $\Si_\La$ is defined as the union over $D$ of the set of functions
$\s:J(D) \to\{1,2,\dots,q\}$ with the property that
$\s_{(x,u-)} \ne \s_{(x,u+)}$ for all $(x,u)\in D$. 
The corresponding probability measure on $\Si_\La$
is found by integrating over $\om$ in the above recipe, as in the following summary.
For $\s\in \Si_\La$, write $D_\s$ for the set
of points $(x,u)\in \La$ such that $\s_{(x,u-) } \ne \s_{(x,u+)}$.
The probability measure $\PPP$ associated with the continuum $q$-state Potts model
on $\La$ is given by
$$
d\PPP(\s) = \frac 1{Z'} e^{\lam L(\s)} d\PP_\de(D_\s),\qquad \s\in\Si_\La,
$$
where $\PP_\de$ is the law of an independent family of
Poisson processes with intensity $\de$ on
the time-lines indexed by $[a,b]$, and
\be\label{eq:new}
L(\s) = \sum_{x\sim y} \int_s^t \de_{\s_{(x,u)}, \s_{(y,u)}}\,du
\ee
is the total length of neighbouring time-lines where the spins are equal.
Here, the summation is over all unordered pairs $x$, $y$
of neighbours. We shall not develop the theory of such measures here, save
for noting that $\PPP$ has the spatial Markov property (see
\cite{bremaud,Grim72} for accounts of the spatial Markov property
for a lattice model).
For $\s\in\Si_\La$ and  a measurable subset $S$ of $\La$, we write $\s_S$ for the
value of $\s$ restricted to $S$, and $\sG_S$ for the $\sigma$-field generated
by $\s_S$. 

The above definition of the continuum \rc\ model is based on an assumption
of free boundary conditions on left/right sides of the region $\La$
(we shall always assume top/bottom periodic conditions in 
this paper). More general boundary conditions 
may be introduced as follows.  Let $\tau$ be an admissible 
configuration of deaths and bridges off the box $\La$. That is, 
$\tau$ comprises a set $D(\tau)$ of deaths and a set $B(\tau)$
of bridges of $(\ZZ\times\RR)\sm\La$ such that: 
the intersection of $D(\tau)$ and $B(\tau)$ with
any bounded sub-interval of $\ZR$ is finite, and no
death is the endpoint of any bridge. For $\om\in\Om_\La$,
we denote by $(\om,\tau)$ the composite configuration comprising $\om$ on $\La$
and $\tau$ on its complement. We write
$\PLlbq^\tau$ for the continuum \rc\ measure
on $\Om_\La$ with the difference that the
number $k(\om)$ of clusters in \eqref{rcPo} is replaced by the number
$k(\om,\tau)$  of clusters of 
$(\om,\tau)$ that intersect $\La$ (subject, as usual, to
the top/bottom periodic boundary condition). As in the lattice case,
$\PLlbq^\tau$ is stochastically increasing in $\tau$.
One may consider also periodic 
boundary conditions.

We extend this discussion now to boundary conditions defined in
terms of spins rather than deaths/bridges. 
Let $q\ge 2$ be an integer. Let $\tau$ be a boundary 
condition as above, and let $\eta$ be a mapping from its clusters to
the set $\{1,2,\dots,q\}$; that is, $\eta$ allocates a spin to each cluster of
$\tau$, viewed as a configuration
on $(\ZR)\sm\La$. Let the measure $\PLlbq^\eta$ be given as $\PLlbq$,
conditioned on the event that no two points $x,y\in\pdh\La$ with
$\eta(x) \ne \eta(y)$ are connected. 
We now allocate spins to the clusters
of the composite configuration $(\om,\tau)$ by: if a cluster $C$
contains a vertex $y$ that is already labelled, the entire cluster 
of $y$ takes that label,
and if no such vertex exists, the spin of $C$ is chosen
uniformly at random from $\{1,2,\dots,q\}$, 
independently of the spins on other clusters.

\section{Basic estimate for the slit box}\label{be}
We consider next a variant of the above model in which the box
$\La$ possesses a `slit' at its centre. Let $L\ge 0$ and
$S_L=[0,L]\times\{0\}$. We think of $S_L$ as a collection
of $L+1$ vertices labelled in the obvious way as
$x=0,1,2,\dots,L$. For $m\ge 2$, $\b>0$, let $\Lamb$ be the box
$[-m,m+L]\times[-\frac12\b,\frac12\b]$ subject to a `slit' along $S_L$. That is,
$\Lamb$ is the usual box except in that each vertex  $x\in S_L$ is
replaced by two distinct vertices $x^+$ and $x^-$. The vertex
$x^+$ (\resp, $x^-$) is attached to the half-line
$\{x\}\times(0,\oo)$ (\resp, the half-line $\{x\}\times(-\oo,0)$);
there is no direct connection between $x^+$ and $x^-$. Write
$S_L^\pm=\{x^\pm: x\in S_L\}$ for the upper and lower sections of
the slit $S_L$. We now construct the continuum \rc\ measure $\phmb$
on $\Lamb$ with top/bottom periodic boundary condition and 
parameters $\lam$, $\de$, $q=2$. We shall abuse
notation by using $\phmb$ to denote also the coupling of the
continuum \rc\ measure and the spin-configuration on $\Lamb$
obtained as above. An illustration of the slit box
is presented in Figure \ref{fig:rdo}.

Let $\Ommb$ be the sample space of the continuum \rc\ model
on $\Lamb$, and $\Smb$ the set of all possible
spin-configurations. That is, $\Smb$ comprises
all admissible allocations of spins to the clusters of
configurations in $\Ommb$. For $\s\in\Smb$ and $x\in S_L$,
write $\s_x^\pm$ for the spin-state of $x^\pm$.
Let $\Si_L=\{-1,+1\}^{L+1}$ be the set of
spin-configurations of the vectors $\{x^+: x\in S_L\}$
and $\{x^-: x\in S_L\}$, and write $\s^+_L=
(\s_x^+: x\in S_L)$ and $\s^-_L=
(\s_x^-: x\in S_L)$. 

It may be checked from \eqref{redstate} that
$$
\phmb (\s_L^-=\eps^-,\, \s_L^+=\eps^+) \propto
\langle \eps^-| \rho_m^L(\b) |\eps^+\rangle,
\qquad \eps^-,\eps^+\in \Si_L,
$$ 
whence $f(\eps^+,\eps^-)=
\phmb(\s_L^-=\eps^-,\, \s_L^+=\eps^+)$ is the function defined in \eqref{eq:1}.
It is easily seen that $\amb$, given in \eqref{eq:amdef}, may be expressed as
\be\label{eq:amdef2}
\amb=\phmb(\s_L^+=\s_L^-).
\ee
On recalling \eqref{sup}, by \eqref{eq:0},
\be\label{eq:4}
\langle \psi| \rho_m^L(\b)-\rho_n^L(\b) |\psi\rangle
= \frac{\phmb(c(\s_L^+)\ol{c(\s_L^-)})}{\amb}
- \frac{\phnb(c(\s_L^+)\ol{c(\s_L^-)})}{\anb} 
\ee
where $c: \Si_L \to \CC$ and
$$
\psi=\sum_{\s_L\in \Si_L} c(\s_L) \s_L \in \sH_L. 
$$

The reduced ground state $\rho_m^L$ is obtained from $\rho_m^L(\b)$ by taking the limit
as $\b\to\oo$. By the remarks in Section \ref{rc},
there exists a probability measure $\phm$ such that
$$
\phmb \Rightarrow \phm \qquad\mbox{as } \b\to\oo.
$$
Furthermore, the $\s_L^\pm$ are cylinder functions, and therefore,
as $\b\to\oo$,
\be\label{eq:34}
\phmb(c(\s_L^+)\ol{c(\s_L^-)}) \to \phm(c(\s_L^+)\ol{c(\s_L^-)}),
\ee
and
\be\label{eq:35}
\amb \to a_m = \phm(\s_L^+=\s_L^-).
\ee

In order to prove Theorem \ref{mainest2},
we seek the function $c: \Si_L\to\CC$, with 
$$
\|c\| = \sqrt{\sum_{\eps\in\Si_L}|c(\eps)|^2} = 1,
$$
that maximises the modulus of \eqref{eq:4}. 
By splitting \eqref{eq:4} into its real and imaginary parts, and
applying the triangle inequality, we see that it suffices to
consider functions $c$ taking only 
non-negative real values. 

Here is the main estimate of this section, of which Theorem \ref{mainest2}
is an immediate corollary with adapted values of the constants.

\begin{thm}\label{mainest}
Let $\lam, \de\in(0,\oo)$ and write $\th=\lam/\de$. 
If $\th < 1$, there exist $\a,C,M \in(0,\oo)$, depending on
$\th$ only, such that
the following holds. There exists $\g=\g(\th)$ satisfying
$\g >0$ when $\th < 1$ such that, for all $L\ge 1$ and $M\le m\le n<\oo$,
\be\label{eq:36}
\sup_{\|c\|=1} \left| \frac{\phm(c(\s_L^+)c(\s_L^-))}{a_m}
- \frac{\pn(c(\s_L^+)c(\s_L^-))}{a_n} \right|
\le C L^\a  e^{-\g m},
\ee
where the supremum is over all functions $c:\Si_L\to\RR$ with $\|c\|=1$. 
The function $\g$ may be chosen
to satisfy
$\g(\th)\to\oo$ as $\th\downarrow 0$.
\end{thm}

The condition $\th < 1$ is important in that it permits a
comparison of the $q=2$ continuum \rc\ model on $\ZR$ with the continuum
percolation model. The claim of the theorem is
presumably valid for $\theta<\thetac$ where
$\thetac$ is the critical point of the former model. [It
may be shown that $\thetac \ge 2$, and we conjecture
that $\thetac=2$, the self-dual point.] Similarly,
Theorem \ref{mainest} has a counterpart in $d \ge 2$
dimensions.

We shall require for the purposes of comparison the following exponential-decay
theorem for continuum percolation. Let $\La_m$ denote the box $[-m,m]^2$, and let 
$I=\{0\}\times[-\frac12,\frac12]$ be a unit `time-segment'
centred at the origin.

\begin{thm}\label{contperc}
Let $\lam, \de \in(0,\oo)$. There exist $C=C(\lam,\de)\in(0,\oo)$
and $\g=\g(\lam,\de)$ satisfying $\g>0$ when $\lam/\de<1$, such that{\rm:}
$$
\Plb\bigl(I\lra \pd\La_m\bigr) \le Ce^{-\g m},\qquad m\ge 0.
$$
The function $\g(\lam,\de)$ may be chosen to satisfy $\g\to\oo$
as $\de\to \oo$ for fixed $\lam$.
\end{thm}

\begin{proof}
Consider the continuum percolation process with parameters $\lam$, $\de$.
The existence of such $\g$ is proved in \cite{BG}.
That $\g\to\oo$ as $\de\downarrow 0$ (with $\lam$ fixed)
may be proved by bounding the cluster at the origin
by a branching process. Consider an age-dependent branching process
in which each particle lives for a length of time having the
distribution of the sum of two independent exponentially-distributed
random variables with parameter $\de$. During its lifetime, it has children
in the manner of a Poisson process with parameter $2\lam$, so that a typical family-size
$N$ has generating function 
$$
G_N(s) = E(s^N) = \left(\frac{\de}{\de-2\lam(s-1)}\right)^2, \qquad |s|\le 1.
$$
The process is subcritical if $E(N)<1$, which is to say that
$G'_N(1) = 4\lam/\de < 1$. When this holds, the 
tail of the total number $M$ of particles decays exponentially, 
and similarly the
aggregate lifetime $U$ of the 
particles has an exponentially-decaying tail. 
See \cite{GS,Har} for accounts of the theory of branching processes.

The branching process dominates $C$ in the following sense. Identify the progenitor
of the branching process and the origin $0$ of $\ZR$. The length
of the maximal death-free time-interval containing the origin has the distribution 
of the lifetime of $0$. The number of bridges with an endpoint in this interval has
the distribution of $N$. Each such bridge has endpoints of the form
$(0,s)$ and $(x,s)$ where $x = \pm 1$. When we iterate this, we find that the number
of bridges in the maximal death-free interval containing 
$(x,s)$ is dominated (stochastically)
by $N$. Arguing inductively, the number of bridges in the cluster $C$ is
dominated stochastically by the total size $M$ of the branching process. 

The horizontal displacement of $C$ is thus smaller (in distribution) than the 
total size $M$ of the branching process. 
It is standard that the tail
of $M$ satisfies $P(M > m) \le Ce^{-\nu m}$ for some $C, \nu
> 0$ depending on $\lam$, $\de$, 
and furthermore that $\nu\to\oo$
if $\de\downarrow 0$ with $\lam$ held fixed.
The behaviour of $\nu$
may be calculated exactly by elementary means, as follows.
One may consider a variant of the branching process in which
each particle has a lifetime with the exponential distribution, parameter
$\de$, and has \emph{pairs} of children at rate $2\lam$ while alive.
The probability generating function of the total progeny
may be found in closed form in the usual way (see \cite{GS}, Problem 5.12.11), 
and one obtains thus a sharp estimate for $\nu$
via Markov's inequality.

Similarly, the vertical displacement of $C$ is smaller (in distribution) than
the aggregate lifetime $U$ of the particles in the branching process. Just as above,
$U$ has exponentially-decaying tail when $E(N)<1$, and the constant in the exponent
tends to infinity as $\de\downarrow 0$ for fixed $\lam$.

Now,
$$
\Plb(0 \lra \pd\La_m) \le P(M \ge m) + P(U \ge m).
$$
A little more is needed for the theorem. The interval $I$ is connected to 
a number of bridges having the Poisson distribution with parameter $2\lam$.
The clusters generated by the ends of these bridges have sizes dominated (stochastically)
as above, and the claim follows.
\end{proof}

In the proof of Theorem \ref{mainest},
we make use of the following two lemmas, which
are proved in the next section using the method of `ratio weak-mixing'.

\begin{lem}\label{lem1}
Let $\lam,\de\in(0,\oo)$ satisfy $\lam/\de<1$. 
There exist constants $\a,C_1,C_2\in(0,\oo)$ such that{\rm:}
for all $L\ge 0$, $m\ge 1$, $\b > 2m+L$, and all $\eps^+,\eps^-\in\Si_L$,
$$
C_1 L^{-\a}
\le \frac{\phmb(\s_L^+=\eps^+,\, \s_L^-=\eps^-)}
{\phmb(\s_L^+=\eps^+)\phmb(\s_L^-=\eps^-)}
\le
C_2 L^{\a}.
$$
\end{lem}

In the second lemma we allow a general boundary condition on $\Lamb$. 

\begin{lem}\label{lem2}
Let $\lam,\de\in(0,\oo)$. 
There exist constants $C,\g\in(0,\oo)$ satisfying $0<\g<1$ when
$\lam/\de<1$ such that{\rm:}
for all $L\ge 0$, $m\ge 1$, $\b\ge 4(m+L+1)$, all events $A\subseteq\Si_L\times\Si_L$,
and all admissible \rc\ boundary-conditions $\tau$ and spin boundary-conditions $\eta$
of $\Lamb$,
$$
\left| \frac{\phmb^{\a}((\s^+_L,\s_L^-)\in A)}
{\phmb((\s^+_L,\s_L^-)\in A)}-1\right|
\le Ce^{-\frac27\g m},\qquad\mbox{for } \a = \tau,\eta,
$$
whenever the right side of the inequality is less than or equal to $1$.
The function $\g$ may be taken as that of Theorem \ref{contperc}.
\end{lem}

The above two lemmas are stated in terms of the box $\Lamb$ with top/bottom 
periodic boundary conditions. 
Their proofs are valid under other boundary conditions also, including
free boundary conditions. 
We make use of this observation during the proofs that follow.

The supremum in Theorem \ref{mainest} may be handled by way of the next lemma.

\begin{lem}\label{supc}
Let $\mu$ be a probability measure on the finite set $S$. Let $\sC$ be the class of functions
$c: S\to [0,\oo)$ such that $\sum_{s\in S} c(s)^2 = 1$. Then
$$
\sum_{s\in S} c(s)\mu(s) \le \sqrt{\sum_{s\in S} \mu(s)^2},\qq c\in\sC,
$$
with equality if and only if 
\begin{equation*}
c(s) = \frac{\mu(s)}{\sqrt{\sum_{t\in S} \mu(t)^2}},
\qquad s\in S.
\end{equation*}
\end{lem}

\begin{proof}[Proof of Lemma \ref{supc}]
This is easily proved using a Lagrange multiplier.
\end{proof}

\begin{proof}[Proof of Theorem \ref{mainest}]
Let $0<\lam<\de$, and let $\g$ be as in Theorem \ref{contperc}.
Let $2 \le m \le n < \oo$ and take $\b > 4(m+L+1)$. Later we 
shall let $\b\to\oo$.
Since $\phmb\lest\phnb$, 
we may couple $\phmb$ and $\phnb$ via a probability measure
$\nu$ on pairs $(\om_1,\om_2)$ of configurations on $\Lanb$ in such a way that
$\nu(\om_1\le\om_2) = 1$. It is standard
(as in \cite{G-RC,New93}) that we may find $\nu$ such that
$\om_1$ and $\om_2$ are identical configurations within the region of $\Lamb$ that
is not connected to $\pdh\Lamb$ in the upper configuration 
$\om_2$.
Let $D$ be the set of all pairs 
$(\om_1,\om_2)\in\Omnb\times\Omnb$ such that: $\om_2$ contains
no path joining $\pd B$ to $\pdh\Lamb$, where
$B = [-r,r+L]\times[-2(r+L+1),2(r+L+1)]$ and $r$ ($< \frac12 m$) will be chosen later. 
We take free boundary conditions on $B$.
 The relevant regions are illustrated in Figure \ref{fig:0}.

\begin{figure}[ht] 
\begin{center}
\input{gosfig0.pstex_t}
\caption{The boxes $\Lanb$, $\Lamb$, and $B$.}\label{fig:0}
\end{center}
\end{figure}

Having constructed the measure $\nu$ accordingly, 
we may now allocate spins to the clusters of
$\om_1$ and $\om_2$ in the manner described earlier. This may be done in such a way that,
on the event $D$, the spin-configurations associated with $\om_1$ and $\om_2$
within $B$ are identical. We write $\s_1$ (\resp, $\s_2$) for
the spin-configuration on the clusters of $\om_1$ (\resp, $\om_2)$, and
$\s_{i,L}^\pm$ for the spins of $\s_i$ on the slit $S_L$.

For $c:\Si_L\to[0,\oo)$ with $\|c\|=1$, let
\be
S_c = \frac{c(\s_{1,L}^+)c(\s_{1,L}^-)}{\amb}
- \frac{c(\s_{2,L}^+)c(\s_{2,L}^-)}{\anb},
\ee
so that
\be\label{e1}
\frac{\phmb(c(\s_L^+)c(\s_L^-))}{\amb}
- \frac{\phnb(c(\s_L^+)c(\s_L^-))}{\anb}
= \nu (S_c;D) + \nu(S_c;\ol D).
\ee
Here, $\ol D$ is the complement of $D$, and
$\nu(f;D)$ denotes $\nu(f1_D)$.

We consider first the term $\nu(S_c;D)$ in \eqref{e1}. On the event
$D$, we have that $\s_{1,L}^\pm = \s_{2,L}^\pm$, so that
\be\label{e2}
|\nu(S_c;D)| \le \left|1-\frac {\amb}{\anb}\right|
\frac{\phmb( c(\s_{L}^+)c(\s_{L}^-))}{\amb}.
\ee
By Lemmas \ref{lem1} and \ref{supc},
\begin{align}
\phmb(c(\s_{L}^+)c(\s_{L}^-))
&= \sum_{\eps^\pm \in \Si_L} c(\eps^+)c(\eps^-) \phmb( \s_{L}^+=\eps^+,\, \s_{L}^-=\eps^- )
\nonumber\\
&\le C_2 L^\a \phmb(c(\s_L^+))\phmb(c(\s_L^-))\nonumber\\
&= C_2 L^\a \left(\sum_{\eps\in \Si_L}c(\eps)\phmb(\s_L^+=\eps) \right)^2\nonumber\\
&\le C_2 L^\a \sum_{\eps\in\Si_L} \phmb(\s_L^+=\eps)^2,\label{eq:300}
\end{align}
where we have used reflection-symmetry in the horizontal axis at the 
intermediate step.
By Lemma \ref{lem1} and reflection-symmetry again,
\begin{align*}
\amb &= \sum_{\eps\in\Si_L} \phmb(\s_L^+=\s_L^-=\eps)\\
&\ge C_1 L^{-\a} \sum_{\eps\in\Si_L} \phmb(\s_L^+=\eps)^2.
\end{align*}
Therefore,
\be\label{e4}
\frac{\phmb( c(\s_{L}^+)c(\s_{L}^-))}{\amb} \le
C_3 L^{2\a},
\ee
where $C_3=C_2/C_1$.

We set $A=\{\s_L^+=\s_L^-\}$ in Lemma \ref{lem2} to find that, for sufficiently
large $m\ge M'(\lam,\de)$,
$$
\left| \frac{\phmb^\eta(\s^+_L=\s_L^-)}
{\phmb(\s^+_L=\s_L^-)}-1\right|
\le C e^{-\frac27 \g m}<\frac 12.
$$
By averaging over $\eta$, sampled according to $\phnb$, we deduce that
$$
\left| \frac{\phnb(\s^+_L=\s_L^-)}
{\phmb(\s^+_L=\s_L^-)}-1\right|
\le Ce^{-\frac27\g m} < \frac12,
$$
which is to say that
\be\label{e3}
\left|\frac{\anb}{\amb}-1\right| \le C e^{-\frac27\g m} < \frac12.
\ee

We make a note for later use.
By the remark after Lemma \ref{lem2}, inequality \eqref{e4} holds also
with $\phmb$ replaced by the continuum \rc\ measure $\phi_B$
on the box $B$ with free boundary conditions. Similarly, 
we may take $C$ and $M'$ above such that
\be\label{e10}
\left|\frac{\anb}{a_B}-1\right| \le C e^{-\frac27\g r} < \frac12, \qquad r\ge M'(\lam,\de),
\ee
where $a_B = \phi_B(\s_L^+ = \s_L^-)$.

Inequalities \eqref{e4} and \eqref{e3} may be combined as in
\eqref{e2} to obtain
\be\label{e5}
|\nu(S_c;D)| \le  C_4 L^{2\a}e^{-\frac27\g m}
\ee
for an appropriate constant $C_4 = C_4(\lam,\de)$ and all $m\ge M'(\lam,\de)$.

We turn to the term $\nu(S_c; \ol D)$ in \eqref{e1}.
Evidently,
\be\label{e6}
|\nu(S_c;\ol D)| \le A_m + B_n,
\ee
where
$$
A_m = \frac {\nu(c(\s_{1,L}^+)c(\s_{1,L}^-);\ol D)}{\amb},\q
B_n = \frac {\nu(c(\s_{2,L}^+)c(\s_{2,L}^-);\ol D)}{\anb}.
$$
There exist constants $C_5$, $M''$ depending on $\lam$, $\de$,
such that, for $m>r \ge M''$,
\begin{align}
B_n &=
\frac{\nu(\ol D)}{\anb} \nu(c(\s_{2,L}^+)c(\s_{2,L}^-)\mid \ol D )\nonumber\\
&= \frac{\nu(\ol D) }{\anb}
\phnb\bigl( \phi_B^\tau(c(\s_{2,L}^+)c(\s_{2,L}^-))\mid \ol D\bigr)
\nonumber\\ 
&\le \frac{\nu(\ol D)}{a_B} C_5 \phi_B(c(\s_{2,L}^+)c(\s_{2,L}^-))
\label{eq:301}
\end{align} 
by Lemma \ref{lem2} with $\phmb$ replaced by $\phi_B$, and \eqref{e10}.
At the middle step, we have used conditional expectation
given the configuration $\tau$
on $\Lamb\sm B$.
By \eqref{e4} applied to the measure $\phi_B$, there exists $C_6 = C_6(\lam,\de)$ such that
\be
\frac 1{a_B}\phi_B(c(\s_{2,L}^+)c(\s_{2,L}^-))  \le C_6 L^{2\a}.
\label{e7}
\ee
Inequalities \eqref{eq:301}--\eqref{e7} imply an upper bound for $B_n$.

A similar upper bound is valid for $A_m$, on noting that the conditioning
on $\ol D$ imparts certain information about the configuration
$\om_1$ outside $B$ but nothing
further about $\om_1$ within $B$.
Combining this with \eqref{e6}--\eqref{e7}, we find that,
for $r \ge M'''(\lam,\de)$ and some $C_7=C_7(\lam,\de)$,
\be\label{e8}
|\nu(S_c;\ol D)| \le \nu(\ol D) C_7 L^{2\a}.
\ee
Let $r=  M'''$ to obtain by \eqref{stochcomp} and Theorem \ref{contperc}
that
\be\label{e23}
\nu(\ol D) \le C_8 (r+L) e^{-\frac12\gamma m} \le C_9 L e^{-\frac12\gamma m} ,
\qquad m \ge 2M''',
\ee 
for some $C_8$, $C_9$.
We combine \eqref{e5}, \eqref{e8}, \eqref{e23} as in \eqref{e1},
and let $\b\to\oo$
to obtain \eqref{eq:36} from \eqref{eq:34}--\eqref{eq:35}, for $m\ge 
\max\{M',M'',2M'''\}$.
The constants $C$, $\g$
may be amended to obtain the required inequality.

Finally, we remark that $\a$, $C$, and $M$ depend on $\lam$ and $\de$.
The left side of \eqref{eq:36} is invariant under re-scalings of the time-axes,
that is, under the transformations $(\lam,\de) \mapsto (\lam\eta, \de\eta)$ 
for $\eta\in(0,\oo)$. We may therefore work with the new values
$\lam'=\th$, $\de'=1$, with appropriate constants $\a(\th,1)$,
$C(\th,1)$, $M(\th,1)$. 
\end{proof}

\section{Ratio weak-mixing}\label{rwm}
Our proofs of Lemmas \ref{lem1} and \ref{lem2} make use of various
couplings of \rc\ measures. Such couplings are fairly standard
(see \cite{G-RC, New93} for example) and have been utilised in
\cite{Al1,Al2} in a study of ratio weak-mixing for \rc\ and spin
models on discrete lattices. We follow in part the arguments of
\cite{Al1,Al2}, but we are not concerned here with the level of
generality of those papers.

Here is some notation. Let $\La$ be a box in $\ZR$ (we shall later
consider a box $\La$ with a slit $S_L$, for which the same
definitions and results are valid).
A {\em path} $\pi$ of $\La$ is an alternating sequence of disjoint intervals
(contained in $\La$) and unit line-segments 
of the form $[z_0,z_1]$, $b_{12}$, $[z_2,z_3]$, $b_{34}$,
$\dots$, $b_{2k-1,2k}$, $[z_{2k},z_{2k+1}]$, where: each pair $z_{2i}$, $z_{2i+1}$ is on the
same `time-line' of $\La$, and $b_{2i-1,2i}$ is a unit line-segment with endpoints $z_{2i-1}$
and $z_{2i}$, perpendicular to the time-lines. 
Note that the equality $z_{2i}=z_{2i+1}$ is permitted.
The path $\pi$ is said to join $z_0$ and $z_{2k+1}$.
The {\em length} of $\pi$ is its
one-dimensional Lebesgue measure, with $\pi$ viewed as a union
of line-segments of $\RR^2$; note that each bridge of
$\pi$ contributes 1 to its length. A \emph{circuit} $D$ of $\La$ is a path 
except inasmuch as
$z_0=z_{2k+1}$. A set $D$ is called \emph{linear}
if it is a disjoint union of paths and/or 
circuits. Let $\De$, $\Ga$ be disjoint subsets
of $\La$. The linear set $D$ is said to \emph{separate} $\De$ and $\Ga$
if every path of $\La$ from $\De$ to $\Ga$ passes through $D$, and $D$ is minimal with
this property in that no strict subset of $D$ has the property.

Let $\om\in\OmL$. An {\em open path} $\pi$ of $\om$ is a path of $\La$
such that, in the notation above,
the intervals $[z_{2i},z_{2i+1}]$ contain no death of $\om$, and the line-segments
$b_{2i-1,2i}$ are bridges of $\om$.

The (one-dimensional) Lebesgue measure of a measurable subset $S$ of $\ZR$ is
denoted $|S|$. Let $S$ and $T$ be measurable subsets of $\La$. The distance $d(S,T)$ from
$S$ to $T$ is defined to be the infimum of the lengths of paths
having one endpoint in
$S$ and one in $T$. Note that the distance function $d$ depends
on the choice of $\La$ (and, in particular, on the boundary conditions
and the presence/absence of a slit).

Let $\phi_\La$ denote the \rc\ measure on $\OmL$
with parameters $\lam$, $\de$, $q=2$ (with top/bottom
periodic boundary condition).
Let $\Ga$ be a measurable subset and $\De$ a finite subset of $\La$ such that
$\De\cap \Ga=\es$.
We shall prove a `ratio weak-mixing property'
of the spin-configurations in $\De$ and $\Ga$.
In order to introduce the necessary couplings,
we consider next a certain `wired' boundary condition on $\La$.
Let $\ophi$ denote the continuum \rc\ measure on $\La$ with parameters
$\lam$, $\de$, $q=2$, but subject to the difference that the set
of clusters that intersect $\De\cup\Ga$ count only $1$
in all towards
the cluster count $k(\om)$ in \eqref{rcPo}. We call $\ophi$ a
`wired \rc\ measure'. It is standard, just as in the
discrete case, that $\ophi$ may be used to generate a random
spin-configuration on $\La$ corresponding to a continuum Ising
model {\em conditioned} on having the same spin at all points in
$\De\cup\Ga$: let $\om$ be sampled according to $\ophi$,
and allocate a randomly chosen spin from the spin set $\{-1,+1\}$
to each cluster of $\om$, these spins being independent between clusters.

Just as in the lattice case, one may use
$\ophi$ to obtain \rc\ measures with other boundary conditions.
Let $\tau\in\Si_{\Ga}$, and let $T_i=\{x\in\Ga: \tau(x)=i\}$
for $i=\pm 1$. The corresponding \rc\ measure, denoted $\phL^\tau$
(as in Section \ref{rc}), is that obtained by: (i)
the set of clusters intersecting $\Ga$ counts
only $1$ in all towards the cluster count in \eqref{rcPo}, and (ii)
we condition on the event that there exists no path joining $T_1$ and $T_2$.
Since $\ophi \gest \phL^\tau$, there exists a coupling $\kappa$ of the two measures with the
property that $\kappa((\om_1,\om_2):\om_1\ge \om_2) = 1$.
It is natural to allocate spins to the clusters of $\om_1$ and $\om_2$
in such a way that, whenever a cluster $C$ of $\om_2$
is also a cluster of $\om_1$, and $C \cap \Ga =\es$,
then these two clusters have the same spin.

One may carry out the above construction 
simultaneously for two (or more) $\tau$.
Let $\tau,\tau'\in \Si_{\Ga}$. 
We may find a coupling of $\ophi$, $\phL^\tau$, $\phL^{\tau'}$
such that the first component is greater than each of the other two. That is,
there exists a measure $\kappa$ on $\Om_\La^3=\{(\om,\om_1,\om_2)\}$ such that:
$\om$ (\resp, $\om_1$, $\om_2$) has law $\ophi$ (\resp, $\phL^\tau$, $\phL^{\tau'}$),
and $\kappa(\om\ge \om_1,\om_2)=1$.

\begin{thm}[Ratio weak-mixing]\label{rweakm}
Let $\Ga\subseteq\La$ be measurable,  let $\De\subseteq \La$ be
finite such that $\De\cap\Ga=\es$, and let $D$ be
a linear subset of $\La$ that separates $\De$ and $\Ga$.
Let $\lam,\de\in(0,\oo)$.
For $\tau,\tau'\in\Si_{\Ga}$ and $\a\in\Si_\De$,
\be\label{thme1}
\left| \frac{\phL^\tau(\s_\De = \a)}{\phL^{\tau'}(\s_\De = \a)} - 1\right|
\le 2\left(t_1+2t_2 +\frac{t_1+t_2}{1-t_1-2t_2}\right),
\ee
whenever the right side is less than or equal to $1$, and where
\begin{equation}\label{eq:121}
t_1=\ol\phi(\De\lra D), \qquad
t_2=\sqrt{\ol\phi(D \lra \Ga)}.
\end{equation}
\end{thm}

The corresponding conclusion is valid when $\La$ is
taken as the slit box $\Lamb$.
Note in this case that the $t_i$ are given in terms of connection probabilities in the
slit box.

\begin{proof}
We adapt the methods of \cite{Al1}. Let $I$ (\resp, $E$)
be the region of $\La$ reachable from $\De$ (\resp, $\Ga$) along paths of $\La$
not intersecting $D$.

Let $\tau,\tau'\in\Si_{\Ga}$ and $\a\in\Si_\De$.
We construct a coupling as follows, using the approach summarised
prior to the statement of the theorem. Let $\ol\om$ have law $\ophi$.
Let $\om=\om^\tau$ and $\om'=\om^{\tau'}$ have laws $\phL^\tau$ 
and $\phi_\La^{\tau'}$, \resp, and be such that
$\om,\om'\le\ol\om$. Furthermore, we construct $\om$
and $\om'$ in such a way that,
if $\ol\om\in E_2=\{D\nlra \Ga\}$, 
then $\ol\om$, $\om$, and $\om'$ are identical
on $D \cup I$.

To the clusters of $\ol\om$, $\om$, $\om'$ we assign spins in the usual manner,
denoted $\ol\s$, $\s$, $\s'$, \resp, such that: 
on the event $E_2$, the functions $\ol\s$, $\s$, $\s'$
are equal on $D\cup I$. For a reason that will be clearer later, we shall
not work with the pair $\s$, $\s'$ of configurations but instead with a pair $\rho$, 
$\rho'$ defined as follows. 
First, we set
$$
\rho_x = \s_x,\ \rho'_x=\s'_x \qquad \mbox{for } x \in D\cup E.
$$
On the event $F=\{\rho_D=\rho'_D\}$, we sample from
the measure $\phL$ given $F$ to obtain a (random) configuration
$\zeta \in \Si_I$, and we set 
$$
\rho_x=\rho_x' = \zeta_x \qquad\mbox{for } x \in I.
$$
On the complement of $F$, we sample $\rho$ (\resp, $\rho'$)
according to the conditional law $\phi_\La^\tau$  given $(\rho_x: x\in D\cup E)$
(\resp, $\phL^{\tau'}$ given $(\rho_x': x\in D\cup E)$).
By the spatial Markov property of the continuum Ising model
alluded to after \eqref{eq:new}, $\rho$ (\resp, $\rho'$) has law $\phL^\tau$
(\resp, $\phL^{\tau'}$), and furthermore:
\be\label{eq:123}
\rho_I=\rho'_I\quad\text{on the event}\quad\{\rho_D=\rho_D'\},
\ee
and
\be\label{eq:127}
\kappa(\rho_D=\rho'_D) = \kappa(\s_D=\s'_D)  \ge \kappa(E_2) = 1-t_2^2,
\ee
where $\kappa$ is the appropriate probability measure,
and $t_2$ is as in \eqref{eq:121}.

Let $H$ be an event satisfying
\be\label{eq:124}
H \subseteq \{\rho_\De=\rho'_\De\}.
\ee
As in \cite{Al1}, if $\kappa(H)>0$,
\begin{align}
\frac{\phL^\tau(\s_\De = \a)}{\phL^{\tau'}(\s_\De = \a)} &=
\frac{\kappa(\rho_\De=\a)}{\kappa(\rho_\De'=\a)}\nonumber\\
&= \frac{\kappa(H\cap\{\rho_\De=\a\})}{\kappa(H\mid \rho_\De=\a)}
\cdot \frac{\kappa(H\mid \rho_\De'=\a)}{\kappa(H\cap \{\rho_\De'=\a\})}
\nonumber\\
&= \frac{ \kappa(H\mid \rho_\De'=\a)} {\kappa(H\mid \rho_\De=\a)}. \label{e12}
\end{align}
It thus suffices, by an elementary argument, to prove
that
\be\label{eq:125}
\kappa(\ol H\mid \rho_\De=\a),\, \kappa(\ol H\mid \rho'_\De=\a) \le t
\ee
where
\be\label{eq:126}
t= t_1+2t_2 +\frac{t_1+t_2}{1-t_1-2t_2}.
\ee
To see this, assume \eqref{eq:125} with $t\le \frac12$. By
\eqref{e12},
$$
1-t \le \frac{\phL^\tau(\s_\De = \a)}{\phL^{\tau'}(\s_\De = \a)} 
\le \frac 1{1-t}.
$$
Now, $1/(1-t) \le 1+ 2t$ since $t \le \frac12$, and \eqref{thme1} follows.

There are four steps in proving \eqref{eq:125}.
Let $\sG_D$ (\resp, $\sG_D'$)
be the $\s$-field generated by $\rho_D$ (\resp, $\rho_D'$).
Firstly, given that $\ol\om \in E_1=\{\De\nlra D\}$, 
the spin-vector $\s_D$ is (conditionally) independent
of $\s_\De$, whence
$$
\bigl| \kappa(\s_D\in A\mid \s_\De=\a) -\kappa(\s_D\in A\mid \s_\De=\a')\bigr|
\le t_1,\quad A \in \sG_D,\ \a'\in\Si_\De,
$$
with $t_1$ as in \eqref{eq:121}. Averaging over $\a'$, we obtain
$$
\bigl| \kappa(\s_D\in A\mid \s_\De=\a) -\kappa(\s_D\in A)\bigr|
\le t_1,
$$
and hence, by the equidistribution of $\s$ and $\rho$,
\be\label{eq:128}
\bigl| \kappa(\rho_D\in A\mid \rho_\De=\a) -\kappa(\rho_D\in A)\bigr|
\le t_1,\qquad A \in \sG_D.
\ee

Secondly, let
$$
g=\kappa(\rho_D \ne \rho'_D\mid \sG_D), \qquad
g'=\kappa(\rho_D \ne \rho'_D\mid \sG_D'),
$$
and, for $a>0$, let $H = H_a$ be given as 
$$
H_a= \{\rho_D = \rho'_D\} \cap \{g \le a\} \cap \{g'\le a\},
$$
where $a$ will be chosen later. It is easily seen by \eqref{eq:123}
that $H_a$ satisfies \eqref{eq:124}.
By Markov's inequality and  \eqref{eq:127},
$$
\kappa(g > a) \le \frac 1a \kappa(g) \le \frac 1a t_2^2,
$$
and therefore, since $\{g>a\} \in \sG_D$,
\begin{align}
\kappa(g>a\mid \rho_\De=\a) &\le \kappa(g>a) + t_1 \quad\text{by \eqref{eq:128}}
\nonumber\\
&\le \frac 1a t_2^2 + t_1.\label{eq:130}
\end{align}
By a similar argument,
\be\label{eq:131}
\kappa(g'>a\mid \rho_\De'=\a) 
\le \frac 1a t_2^2 + t_1.
\ee

Thirdly,
\begin{align} \kappa(\rho_D \ne \rho_D',\, g\le a \mid \rho_\De=\a)
&\le \esssup\bigl\{\kappa(\rho_D\ne \rho_D'\mid \sG_D)1_{\{g\le a\}}\bigr\}
\nonumber\\
&= \esssup\{g 1_{\{g\le a\}}\}  \le a,
\label{eq:132}
\end{align}
and similarly,
\be\label{eq:133}
\kappa(\rho_D \ne \rho_D',\, g'\le a \mid \rho_\De'=\a)
\le a.
\ee

Finally, by \eqref{eq:123},
\be\label{eq:134}
\{\rho_D=\rho_D'\} \cap \{\rho_\De=\a\} =
\{\rho_D=\rho_D'\} \cap \{\rho_\De'=\a\},
\ee
[this is where we use $\rho$, $\rho'$ in place of $\s$, $\s'$],
and, by \eqref{eq:131} and \eqref{eq:133}--\eqref{eq:134},
\begin{align}\label{eq:135}
\kappa(\rho_D = \rho_D',\, g'>a \mid \rho_\De=\a)
&\le \kappa(g'>a\mid \rho_D=\rho_D',\, \rho_\De'=\a) 
    \nonumber\\
&\le \frac{\kappa(g'>a \mid \rho_\De'=\a)}
          {\kappa(\rho_D=\rho_D' \mid \rho_\De'=\a)}\nonumber\\
&\le \frac{t_1 + t_2^2/a}{1-a-t_1-t_2^2/a}.
\end{align}

On combining \eqref{eq:130}, \eqref{eq:132}, \eqref{eq:135},
and setting $a=t_2$, we obtain the first inequality of
\eqref{eq:125} with $H=H_a$, and the second inequality holds similarly.
\end{proof}

Let $\De$ and $\Ga$ be disjoint finite subsets of $\La$ that are
disjoint from $\pdh\La$. Let $D$ be an linear subset of $\La$ that separates
$\De$ and $\Ga\cup\pdh\La$. Let $\a\in \Si_\De$, $\b,\b'\in\Si_\Ga$, and
$\eta\in\Si_{\pdh\La}$. By \eqref{thme1} applied to the sets $\De$ and $\Ga\cup\pdh\La$,
\be\label{eq:200}
\bigl|\phL^{\b,\eta}(\s_\De=\a) - \phL^{\b',\eta}(\s_\De=\a)\bigr|
\le 2t \phL^{\b',\eta}(\s_\De=\a),
\ee
whenever $t\le \frac12$ where
\be\label{eq:tdef}
t=t_1+2t_2 +\frac{t_1+t_2}{1-t_1-2t_2},
\ee
and
\begin{equation}\label{eq:136}
t_1=\ol\phi(\De\lra D), \qquad
t_2=\sqrt{\ol\phi(D \lra \Ga\cup\pdh\La)}.
\end{equation}
The suffix $\b,\eta$ in \eqref{eq:200} indicates the composite
boundary condition taking the values $\b$ on $\Ga$ and $\eta$ on $\pdh\La$.
We average \eqref{eq:200} over $\b'$ to obtain
\be\label{eq:201}
\bigl|\phL^{\b,\eta}(\s_\De=\a) - \phL^{\eta}(\s_\De=\a)\bigr|
\le 2t \phL^{\eta}(\s_\De=\a).
\ee
Now,
$$
\phL^{\b,\eta}(\s_\De=\a) = \frac {\phL^\eta(\s_\De=\a,\, \s_\Ga=\b)}
{\phL^\eta(\s_\Ga=\b)}, 
$$
Let $A\in \sG_\De$, $B\in \sG_\Ga$ be events with strictly
positive probabilities.
We `multiply up' in  \eqref{eq:201} and sum over $\a\in A$ and $\b\in B$ to
find that
\be\label{thme2}
\left| \frac{\phL^\eta(A\cap B)}{\phL^{\eta}(A)\phL^\eta(B)} - 1\right|
\le
2 t,\qquad \eta\in\Si_{\pdh\La},
\ee
whenever $t\le \frac12$.
Upper bounds on $t$ follow from the observation that
$\ophi$ is stochastically dominated by the continuum
percolation measure with parameters $\lam$, $\de$ (cf.\ \eqref{stochcomp}).
Equation \eqref{thme2} is a general statement of so-called ratio weak-mixing.

By the same argument without the reference to the boundary $\pdh\La$,
\be\label{thme22}
\left| \frac{\phL(A\cap B)}{\phL(A)\phL(B)} - 1\right|
\le
2 t,\qquad A \in \sG_\De,\ B\in\sG_\Ga,
\ee
whenever $t\le \frac12$, where $t$ is in \eqref{eq:tdef} with
\begin{equation}\label{eq:136a}
t_1=\ol\phi(\De\lra D), \qquad
t_2=\sqrt{\ol\phi(D \lra \Ga)},
\end{equation}
and $D$ is a linear set that separates $\De$ and $\Ga$.

The above ideas may be used to prove Lemmas \ref{lem1} and \ref{lem2},
for the first of which we argue as follows. 
Consider the box $\Lamb$ with slit $S_L$. Let $K$ be an integer satisfying
$0< K < \frac12 L$, and let $\De= \{x^+: x\in S_L,\ K\le x\le L-K\}$ and
$\Ga = \{x^-: x\in S_L,\ K\le x\le L-K\}$. 

\begin{lem}\label{thm2}
Let $\lam,\de\in(0,\oo)$. There exists  
$C=C(\lam,\de)\in(0,\oo)$ such that,
if $\b > 2m+L$,
$$
\left| \frac{\phmb(\s_\De=\eps_K^+,\, \s_\Ga=\eps_K^-)}
{\phmb(\s_\De=\eps_K^+)\phmb(\s_\Ga=\eps^-_K)} - 1\right|
\le
Ce^{-\frac12\g K}, \qquad \eps_K^+\in\Si_\De,\ 
\eps_K^-\in\Si_\Ga,
$$
whenever the right side is less
than or equal to $1$. The function $\g(\lam,\de)$ may be taken as that
in Theorem \ref{contperc}.
\end{lem}

The proofs are preceded by a type of `finite-energy' inequality
(see \cite{Al1,G-RC}). 

\begin{lem}\label{finite-energy}
Let $S$ be a finite subset of $\La$.
For $x\in \La \sm S$, $\eps\in\Si_S= \{-1,+1\}^S$, and $\a \in \{-1,+1\}$,
\be
\phL(\s_S=\eps,\, \s_x=\a) \ge \tfrac12 \phL(\s_S=\eps) \PLlb(x \nlra S) .
\ee
\end{lem}

\begin{proof}
Let $x \in S$, $\eps\in\Si_S$, and $\a\in\{-1,+1\}$. 
Let $E(\eps)$ be the decreasing event
containing all $\om\in\Om_\La$ such that: for all $s,t\in S$, $s\nlra t$
whenever $\eps_s\ne \eps_t$. Recalling the manner in which spins
are associated with clusters,
\be\label{e34}
\phL(\s_S=\eps) = \phL(2^{-k(S)} 1_{E(\eps)}), \qquad \eps\in\Si_S,
\ee
where $k(S)$ is the number of clusters intersecting $S$.
Similarly,
\be\label{e36}
\phL(\s_S=\eps,\,\s_x=\a) \ge \phL(2^{-k(S^+)} 1_{E(\eps)}1_{x\nlra S}),
\ee
where $S^+ = S\cup\{x\}$. Note that $k(S^+) = k(S)+1$
when $x \nlra S$.

For any event $A$,
\be\label{e35}
\phL(2^{-k(S^+)}1_A) = \phL(2^{-k(S^+)}) \wh\phi(A)  =K\wh\phi(A),
\ee
where $K=\phL(2^{-k(S^+)})$ and $\wh\phi$ 
is the continuum \rc\ measure on $\La$
with a wired boundary condition on $S^+$, that is, all clusters
intersecting $S^+$ are counted as one.
By \eqref{e35} and the FKG
inequality applied to $\wh\phi$,
\begin{align*}
\phL(2^{-k(S^+)}1_{E(\eps)}1_{x\nlra S}) 
&= K\wh\phi(E(\eps)\cap \{x\nlra S\})\\
&\ge K\wh\phi(E(\eps))\wh\phi(x\nlra S)\\
&= \phL(2^{-k(S^+)}1_{E(\eps)}) \wh\phi(x\nlra S).
\end{align*}
Now $k(S)\le k(S^+) \le k(S) + 1$, so that, by \eqref{e34}--\eqref{e36},
$$
\phL(\s_S=\eps,\, \s_x=\a) 
\ge \tfrac12 \phL(\s_S=\eps) \wh\phi(x\nlra S)
$$
and the claim follows by the stochastic inequality \eqref{stochcomp}.
\end{proof}

\begin{proof}[Proof of Lemma \ref{thm2}] 
Take $D=\{(x,0): x\in[-m,0)\cup(L,L+m]\}$, the union of the two horizontal
line-segments that, when taken with the slit $S_L$, 
complete the `equator' of $\Lamb$. Thus, $D$ is a linear subset of $\Lamb$
separating $\De$ and $\Ga$.
Since $\ol\phi\lest \PLlb$, by Theorem
\ref{contperc} there exist constants $C$, $C'$ depending
on $\lam$ and $\de$ alone, such that
$$
t_1 =\ol\phi(\De\lra D) \le 2\sum_{i=K}^{\lfloor L/2\rfloor} Ce^{-\g i}
\le C' e^{-\g K},
$$
and furthermore $t_2^2=t_1$.
The claim now follows by \eqref{thme22}.
\end{proof}

\begin{proof} [Proof of Lemma \ref{lem1}]
Let $\g$ be given as in Theorem \ref{contperc}. With $K=\lceil \ln L\rceil$,
let $\s_{L,K}^\pm=(\s_x^\pm: K \le x \le L-K)$.
We may apply Lemma \ref{finite-energy} as follows
in order to compare the laws of
the spin-vector $\s_L^\pm$ and that of the reduced vector $\s_{L,K}^\pm$.
First, let $x=(L,0)$, and let $\eps^+, \eps^-\in \{-1,+1\}^{L+1}$ be possible
spin-vectors of the sets $S_L^+$ and  $S_L^-$, respectively. By
Lemma \ref{finite-energy} with $S = S_L^+ \cup S_L^- \sm\{x^+\}$,
\begin{multline*}
\phmb(\s_L^+ = \eps^+,\, \s_L^- =\eps^-)\\
\ge \tfrac12 \phmb(\s_y^+=\eps_y^+\ \text{for}\ y\in S_L^+\sm\{x^+\},\, \s_L^-=\eps^-) 
\PLmblb(x^+ \nlra S).
\end{multline*} 
Now, $\PLmblb(x\nlra S)$ is at least as large as the probability that the
first event (death or bridge) encountered on moving northwards
from $x$ is a death. That is,
$$
\PLmblb(x \nlra S) \ge \frac{\de}{2\lam + \de}.
$$
On iterating the above argument, we obtain that
\be
\phmb(\s_L^+=\eps^+,\, \s_L^-=\eps^-) \ge
\left(\frac{\de}{2(2\lam+\de)}\right)^{4K}
\phmb(\s_{L,K}^+ = \eps_K^+,\, \s_{L,K}^- =\eps_K^-),
\ee
where $\eps^\pm_K$ is the vector obtained from $\eps^\pm$
by removing the entries labelled by vertices $x$ satisfying
$0\le x <K$ and $L-K < x \le L$.
In summary, there exist $C,\a\in(0,\oo)$ depending on $\lam$, $\de$
such that, for $\eps^\pm\in\Si_L$,
\begin{align*}
C L^{-2\a} \phmb(\s_{L,K}^+=\eps_K^+,\, \s_{L,K}^-=\eps_K^-)
&\le \phmb(\s_L^+=\eps^+,\ \s_L^-=\eps^-)\\
&\le  \phmb(\s_{L,K}^+=\eps_K^+,\, \s_{L,K}^-=\eps_K^-).
\end{align*}

Set $\De= \{x^+: x\in S_L,\ K\le x\le L-K\}$,
$\Ga = \{x^-: x\in S_L,\ K\le x\le L-K\}$, and apply Lemma \ref{thm2}
to obtain that
there exists $C=C(\lam,\de)<\oo$ such that
$$
\left|\frac{\phmb(\s_{L,K}^+=\eps_K^+,\, \s_{L,K}^-=\eps_K^-)}
{\phmb(\s_{L,K}^+=\eps_K^+)\phmb(\s_{L,K}^-=\eps_K^-)}-1\right|
\le Ce^{-\frac12\g K}
\le C L^{-\frac12 \g},
$$
whenever (say) the right side is less than or equal to $\frac12$, 
say for $L\ge L_0(\lam,\de)$.

By Lemma \ref{finite-energy} again, for suitable
$C'$, $\a$,
$$
C' L^{-\a} \phmb(\s_{L,K}^\pm=\eps_K^\pm)
\le \phmb(\s_L^\pm=\eps^\pm)
\le  \phmb(\s_{L,K}^\pm=\eps_K^\pm).
$$
The claim now follows for $L\ge L_0$, with suitable values of $C_1$, $C_2$, $\a$.
We may adjust the constants to obtain the required inequality for all $L\ge 0$.
\end{proof}

\begin{proof}[Proof of Lemma \ref{lem2}]
Let $\De=S_L^+ \cup S_L^-$ and $\Ga=\pdh\Lamb$. Let $k=\frac37 m$ and
assume for simplicity that $k$ is an integer. [If either $m$ is small or $k$ is
non-integral, the constant $C$ may be adjusted accordingly.]
Let $D$ be the
circuit illustrated in Figure \ref{fig:4}, comprising a path in
the upper half-plane from $(-k,0)$ to $(L+k,0)$ together with its reflection
in the $x$-axis.

\begin{figure}[ht] 
\psfrag{x}[][]{$x_2$}
\psfrag{y}[][]{$x_3$}
\begin{center}
\input{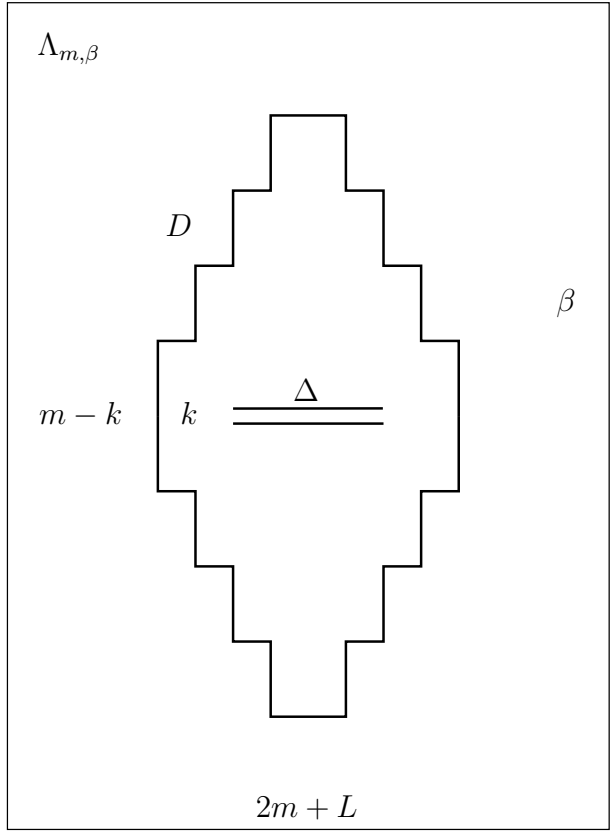}
\caption{The circuit $D$ is approximately a parallelogram with $\De$
at its centre. The sides comprise vertical steps of 
height $2$ followed by horizontal
steps of length 1. The horizontal and vertical diagonals have lengths $2k+L$ and
(order) $4k+2L$ respectively, where $k=\frac37 m$.}\label{fig:4}
\end{center}
\end{figure}

By Theorem \ref{rweakm},
$$
\left| \frac{\phmb^{\a}((\s^+_L,\s_L^-)=(\eps^+,\eps^-))}
{\phmb((\s^+_L,\s_L^-)=(\eps^+,\eps^-))}-1\right|
\le 2t,\qquad \a=\eta,\tau,\ \eps^\pm\in\Si_L,
$$
whenever $t\le \frac12$, with $t$ as in \eqref{eq:tdef}.
We `multiply up' and sum over $(\eps^+,\eps^-)\in A$
to obtain 
\be\label{eq:203}
\left| \frac{\phmb^{\a}(\s_\De\in A)}
{\phmb(\s_\De\in A)}-1\right|
\le 2t,
\ee
whenever $t \le \frac12$.

By \eqref{stochcomp}, $\ol\phi \lest \PLlb$. 
Let $\b \ge 4(m+L+1)$. It is a straightforward 
consequence of Theorem \ref{contperc}
that there exist $C,C',c'>0$, depending on $\lam$, $\de$ only, such that
\be\label{eq:2001}
t_1 \le 4\sum_{i=0}^{\lfloor L/2\rfloor} \Plb((i,0)\lra D)
\le 4\sum_{i=0}^{\lfloor L/2\rfloor}  C e^{-\g\frac23( k +i)} \le C' e^{-\frac27\g m},
\ee
and similarly,
\be\label{eq:2002}
t_2^2 \le 8\sum_{i=0}^{\lceil k+L/2\rceil} C e^{-\g(\frac47m+ c'i)} \le C' e^{-\frac47\g m},
\ee
with $\g$ given as in Theorem \ref{contperc}.
The claim of the lemma follows.
\end{proof}

\section{Disordered interactions}\label{sec:disorder}
We have so far assumed that the spin-couplings
$\lam_{x,x+1}$ and the field-strengths $\de_x$ appearing
in the Hamiltonian \eqref{ham} are constant. The situation
is more complicated if: either the environment of couplings
and strengths vary about the space $\ZZ$, or they are random
(in which case the model is said to be disordered). The
arguments of this paper may be applied in each case, and the outcomes 
are summarised in this section.

Suppose first that the $\lam_{x,x+1}$ and $\de_x$ are non-constant. 
The fundamental bound of Theorem \ref{mainest} depends only on
the ratio $\th=\lam/\de$, and the connection probabilities
of the continuum \rc\ model are increasing in the $\lam_{x,x+1}$
and decreasing in the $\de_x$. One may therefore check that
the conclusions of the paper are valid with $\g=\g(\lam,\de)$ whenever
\be\label{eq:1005}
\lam_{x,y}/\de_x \le \lam/\de,\qquad y=x-1,x+1, \ x \in\ZZ.
\ee
Hence, in the disordered case where \eqref{eq:1005}
holds with probability one, the corresponding conclusion is valid. 

We turn to the situation in which \eqref{eq:1005} does not
hold with probability one. Suppose that the $\lam_{x,x+1}$, $x\in\ZZ$,
are independent, identically distributed random variables, and
similarly the $\de_z$, $z\in\ZZ$, and assume that the $\lam_{x,y}$
are independent of the $\de_z$. We write $P$ for the corresponding probability
measure, viewed as the measure governing the `random environment',
and $\Lambda$, $\De$ for a typical spin-correlation and field-strength, \resp.
Conditional on  $\lam=(\lam_{x,x+1}:x\in\ZZ)$ and $\de=(\de_z:z\in\ZZ)$,
we write $\Plb$ for the probability measure of
the associated continuum percolation process.
In applying the methods of this paper within the random environment, 
one needs to deal with sub-domains of 
$\ZZ$ where the environment is not propitious for the
bound of Theorem \ref{mainest}. As before, we perform a comparison 
of the continuum \rc\ model and continuum percolation in a random
environment, and we shall appeal to the
following theorem of \cite{Klein1} (see also Theorem 1.6
of \cite{AKN}).

For $(x,s), (y,t) \in \ZZ\times \RR$ and $q\ge 1$, let
$$
d_q(x,s;y,t) = \max\bigl\{|x-y|, (\ln^+ |s-t|)^q\bigr\},
$$
where $\ln^+ x = \max\{\ln x, 0\}$.

\begin{thm}\label{thm:klein}{\bf\cite{Klein1}}
Consider continuum percolation on $\ZR$ in a random environment satisfying
$$
\Gamma = \max\left\{P\bigl([\ln(1+\Lambda)]^\b\bigr), 
P\bigl( [\ln(1+\Delta^{-1})]^\b\bigr) \right\} <\oo,
$$
for some 
\be\label{eq:beta}
\b> 5 + \tfrac 72\sqrt{2}.
\ee
There exists $Q=Q(\b)>1$ such that the following holds. For $q\in[1,Q)$ and $\g>0$,
there exists $\eps=\eps(\b,\Gamma,\g,q)>0$ and $\eta=\eta(\b,q)>1$ such that{\rm:} if
\be\label{eq:1010}
P\Bigl(\bigl[\ln(1+(\Lambda/\Delta))\bigr]^\b\Bigr) < \eps,
\ee
there exist identically distributed, positive random variables 
$D_x\in L^\eta(P)$, $x\in\ZZ$, such that
\be\label{eq:1006}
\Plb\bigl((x,s) \lra (y,t)\bigr)\le 
\exp\bigl[-\g d_q(x,s;y,t)\bigr] \quad \text{if } d_q(x,s;y,t)\ge D_x,
\ee
for $(x,s),(y,t) \in \ZZ\times\RR$. 
\end{thm}  

The lower bound \eqref{eq:beta} for $\b$ is enough to imply that $P(D_x^\eta)<\oo$
for some $\eta > 1$. The larger $\b$, the larger $\eta$ may be taken.

For the remainder of this section we assume that 
the conditions of the above theorem are
valid, and we shall work with the conclusion \eqref{eq:1006}, with 
$q=1$, $\g>1$, and the $D_x$ given accordingly.
We let $L\ge 8$ and $K=\lceil \ln L\rceil$, and consider
the event
$$
A_L = \bigcap_{x=K}^{L-K} \bigl\{D_x < \min\{x,L-x\}\bigr\},
$$
noting that
$$
P(A_L) \ge 1- 2 \sum_{x=K}^\oo P(D \ge x),
$$
where $D$ has the distribution of the $D_x$. Since $P(D)<\oo$,
\be\label{eq:1007}
P(A_L) \to 1\qquad\text{as } L\to\oo.
\ee
An estimate for the rate of convergence 
may be obtained (here and later)
by the fact that $P(D^\eta)<\oo$ for some $\eta > 1$.

We comment next on the adaptation of our earlier results to the disordered
setting. 
Theorem \ref{rweakm} holds within the random environment, without change.
The conclusion of Lemma \ref{thm2} is valid with $K=\lceil\ln L\rceil$
whenever the event $A_L$ occurs.
Lemma \ref{finite-energy} holds unconditionally.  The conclusion
of Lemma \ref{lem1} holds on $A_L$ with the lower bound $C_1L^{-\a}$ replaced 
by $C X_{L}$ and the upper bound $C_2L^\a$ replaced by $(CX_{L})^{-1}$,
with $C$ a constant and
$$
X_{L} = \prod_{x\in \Theta} \frac{\de_x}{\de_x + \lam_{x,x-1} + \lam_{x,x+1}},
$$
where, in the notation of the proof of Lemma \ref{lem1},
$\Theta = (S_L^+\sm \De)\cup(S_L^-\sm\Ga)$. 
Now, 
\be\label{eq:1009}
\ln X_{L} = 
-2 \sum_{x=0}^{K-1}  Z_x
-2 \sum_{x=L-K+1}^{L}  Z_x
\ee
where
$$
Z_x
= \ln\left(1+\frac{\lam_{x,x-1} + \lam_{x,x+1}}{\de_x}\right).
$$ 
The two summations in \eqref{eq:1009} are independent of one another, 
and each is the sum
of a $1$-dependent sequence of random variables. Also,
$$
Z_x \le \ln  \left(1+\frac{\lam_{x,x-1}}{\de_x}\right)
+ \ln\left(1+\frac{\lam_{x,x+1}}{\de_x}\right),
$$
so that, by \eqref{eq:1010} and the Minkowski inequality,
$$
\sqrt{P(Z_x^2)} \le 2\sqrt{P\Bigl(\bigl[\ln(1+(\Lambda/\Delta))\bigr]^2\Bigr)}
< \oo.
$$
By the central limit theorem for $1$-dependent
sequences (see, for example, Theorem 19.2.1 of \cite{IL}),
\be\label{eq:1012}
P(B_L^\rho) \to 1\qquad\text{ as } L \to\oo,
\ee
where $B_L^\rho = \{X_{L} \ge L^{-\rho}\}$ and $\rho\in(0,\oo)$ satisfies
\be\label{eq:1011}
\rho> 4P(Z_0).
\ee

Some changes
are necessary to the proof of Lemma
\ref{lem2}, reflecting the fact that the decay in \eqref{eq:1006}
is sub-exponential in time.  The circuit illustrated
in Figure \ref{fig:4} is generated by translation, 
discretisation, and reflection
of the Cartesian line $y=2x$. 
In the disordered setting, we work instead with the curve 
$y=e^{x}$, and we assume $\b > 5e^{m+\frac12 L}$.
We define two further events that depend on the environment.
Assume for simplicity that $m$ is even, write $k=\frac12 m$, and let
\begin{align*}
C_{L,m} &= 
   \bigcap_{x=0}^L \bigl\{D_x < \tfrac 12\min\{k+x,L+k-x\}\bigr\},\\
D_{L,m} &=
   \bigcap_{x=-k}^{L+k} \bigl\{D_x < \min\{m+x,L+m-x\}\bigr\}.
\end{align*}
In the current setting, \eqref{eq:2001}
becomes
$$
t_1 \le C_1 e^{-\frac 14 \g m}\quad\text{on the event} \quad C_{L,m},
$$
for some constant $C_1$ depending on $\g$. 
Similarly, \eqref{eq:2002} is replaced by
$$
t_2^2 \le C_2e^{-\frac 12 \g m} \quad\text{on the event} \quad D_{L,m}.
$$
An amended version of Lemma \ref{lem2} thus holds, so long as the event
$C_{L,m} \cap D_{L,m}$ occurs.

We estimate $P(C_{L,m}\cap D_{L,m})$ as follows.
First, since $P(D)<\oo$,
\be\label{eq:2003}
P(C_{L,m}) \ge 1- 2 \sum_{x=0}^{\lfloor\frac12 L\rfloor}
P(D_x \ge \tfrac12 (k+x)) 
\to 1\qquad\text{as } m\to\oo.
\ee
Similarly,
\be\label{eq:2004}
P(D_{L,m}) \ge 1- 2\sum_{x=-k}^{\lfloor\frac12 L\rfloor}
P(D_x \ge m+x)
\to 1\qquad\text{as } m\to\oo.
\ee

Suppose that $A_L\cap B_L^\rho \cap C_{L,m}\cap D_{L,m}$ occurs for 
some $\rho$ satisfying \eqref{eq:1011}.
The principal estimate \eqref{eq:rcbound} follows 
with $CL^{\a}$ replaced by $C L^{\rho}$ as above.
On the above event, the proof of Theorem \ref{entest} may be followed to
obtain the logarithmic decay of entanglement.
Note from \eqref{eq:1007} and \eqref{eq:1012}
that $P(A_L\cap B_L^\rho) \to 1$ as $L\to\oo$, and by
\eqref{eq:2003}--\eqref{eq:2004} that $P(C_{L,m}\cap D_{L,m})
\to 1$ as $m\to\oo$.

\begin{proof}[Proof of Theorem \ref{thm:klein}]
This is essentially Theorem 1.1 of \cite{Klein1} 
with $d=1$, subject to two differences:
the right side of \eqref{eq:1006} is expressed differently in \cite{Klein1},
and the condition on $\b$ is different.
The present statement is obtained as follows from the proof of \cite{Klein1}, using the
notation of that proof. With $\b$ satisfying \eqref{eq:beta} and $\a=1+\sqrt 2$, we pick
$p>2\a$ and $\nu=q^{-1}$ satisfying (3.3) of \cite{Klein1}.
Let $K_x$ denote the minimal $k_1$ in the
second paragraph of the proof of Theorem 3.3 of \cite{Klein1}. As there,
$$
P(K_x > r) \le \frac c{L_r^{p-\a}},\qquad r \ge 1,
$$
where  $c$ is a constant, and $(L_r:r\ge 1)$ is
a sequence of positive reals given by $L_r=L^{\a^r}$ for some 
large $L$. Let $D_x = bL_{K_x}$. Inequality
\eqref{eq:1006} holds by the argument of \cite{Klein1}.
Furthermore, for $\eta>1$,
$P(D_x^{\eta})$ has the same order as
\be
b^{\eta} \sum_{a = 0}^\oo a^{\eta-1} P(L_{K_x} > a)
\le b^{\eta} \sum_{r= 0}^\oo L^{\eta\a^{r+1}} \cdot\frac 1{L^{(p-\a)\a^r}},
\ee
which is finite whenever $\eta-1$ is small and positive. 
\end{proof}

\bibliography{qent}
\bibliographystyle{plain}

\end{document}